\pdfoutput=1

\documentclass[journal]{IEEEtran}

\usepackage{color}
\usepackage[cmex10]{amsmath}
\usepackage{graphicx}
\usepackage{amsthm}
\usepackage{amssymb}
\usepackage{multicol}
\usepackage{cite}
\usepackage{algpseudocode}
\usepackage{algorithm}
\usepackage{amsfonts}
\usepackage{latexsym,psfrag}
\usepackage{makeidx}
\usepackage{todonotes}
\usepackage{comment}
\usepackage{dsfont}
\usepackage{multirow}
\usepackage{url}
\usepackage{xcolor}
\usepackage{booktabs}
\usepackage{ifpdf}
\usepackage{soul}

\usepackage{tabularx}
\usepackage{booktabs}

\usepackage[moderate]{savetrees}

\newtheorem{theorem}{Theorem}

\newtheorem{claim}{Claim}

\newcommand{\Comps}{\mathbb{C}}

\newcommand{\rank}{\mathrm{rank}\,}
\newcommand{\diag}{\mathrm{diag}\,}

\def\be{\begin{equation}}
\def\ee{\end{equation}}

\def\bearn{\begin{eqnarray*}}
\def\eearn{\end{eqnarray*}}
\def\bear{\begin{eqnarray}}
\def\eear{\end{eqnarray}}
\def\barr{\begin{array}}
\def\earr{\end{array}}

\newcommand{\jylb}[1]{\textcolor{red}{#1}}
\renewcommand{\jylb}[1]{\textcolor{black}{#1}}

\newcommand{\md}[1]{\textcolor{red}{#1}}
\renewcommand{\md}[1]{\textcolor{black}{#1}}

\def\bmat{\left(\begin{array}}
\def\emat{\end{array}\right)}

\newtheorem{lem}{Lemma}

\usepackage{etoolbox}\AtBeginEnvironment{algorithmic}{\small}

\begin{document}

\title{Security Measures for Grids against Rank-1 Undetectable Time-Synchronization Attacks}
\author{Marguerite Delcourt \textit{Student Member, IEEE}, Jean-Yves Le Boudec \textit{Fellow, IEEE}
\thanks{M. Delcourt and J-Y. Le Boudec are with the School of Computer and Communication Sciences of the Swiss Federal Institute of Technology Lausanne, EPFL, Switzerland. e-mail: \{marguerite.delcourt$\vert$ jean-yves.leboudec\}@epfl.ch.}}
\maketitle

\begin{abstract}
Time-synchronization attacks on phasor measurement units (PMU) pose a real threat to smart grids; it was shown that they are feasible in practice and that they can have a non-negligible negative impact on the state estimation, without triggering the bad-data detection mechanisms. Previous works identified vulnerability conditions when targeted PMUs measure a single phasor. Yet, PMUs are capable of measuring several quantities. We present novel vulnerability conditions in the general case where PMUs measure any number of phasors and can share the same time reference. One is a sufficient condition that does not depend on the measurement values. We propose a security requirement that prevents it and provide a greedy offline algorithm that enforces it. If this security requirement is satisfied, there is still a possibility that the grid can be attacked, although we conjecture that it is very unlikely. We identify two sufficient and necessary vulnerability conditions which depend on the measurement values. For each, we provide a metric that shows the distance between the observed and vulnerability conditions. We recommend their monitoring for security. Numerical results, on the IEEE-39 bus benchmark with real load profiles, show that the measurements of a grid satisfying our security requirement are far from vulnerable.
\end{abstract}

\section{Introduction}

Accurate system estimation is a crucial element for the control and operation of smart grids. Such an estimation relies on measurements of electrical quantities, taken at various grid locations. The advancing technology of phasor-measurement units (PMU) enables the use of synchrophasors with fast streaming rates~\cite{RP14,TSG13}. The precision of the estimation from phasors depends on the time synchronization of the PMUs~\cite{BFMP15}, which can be achieved through either a space-based protocol such as GPS synchronization~\cite{CLMS08} or a network-based protocol~\cite{LBFMPM12} such as White Rabbit~\cite{DRWP20}. In both cases, it is vulnerable to cyber attacks~\cite{SBDSF19}. The GPS synchronization of PMUs is vulnerable to GPS spoofing attacks~\cite{Jiang2013GPS} and the network-based protocols are vulnerable to the insertion of a delay box between the PMUs and their master clocks~\cite{Barreto2016}. Network-based protocols assume that the transmission time of a packet between a PMU and its master clock is symmetric (i.e. the time to send packets is the same in both directions), or has a known asymmetry. A delay box breaks this symmetry and introduces a time offset in the time reference of the targeted PMU. These are physical attacks that are neither prevented nor detected by the cryptographic tools used by the synchronization protocol. As a result, the phase of phasor measurements is shifted, which can non-negligibly impact the state estimation of the system, as it was shown in~\cite{rank1,SDBDBP19}.

In order to make the state-estimation process robust, it is customary to couple the state-estimation algorithm with a bad-data detection (BDD) algorithm such as the largest normalized residual test or the chi-squared test~\cite{abur2004power}. However, it was shown in~\cite{Liu2009} that false-data injection attacks can have a non-negligible impact on the state estimation, without triggering any reaction from the BDD algorithms. Subsequent works focused on undetectable attack strategies~\cite{Kosut2011}, on vulnerability identification~\cite{Dan2010,Anwar2015} and on the mitigation of such attacks~\cite{Vukovic2012,5751206}. 

Another type of attack that can impact the state estimation without triggering any reaction from the BDD algorithms, is to alter the time reference of PMUs. Unlike with false-data injection attacks, the data is never forged nor modified. This can be done remotely via GPS spoofing~\cite{Zhang2013} or with delay boxes~\cite{Barreto2016}. The authors of~\cite{rank1} propose a strategy to compute undetectable time-synchronization attacks (TSAs) against vulnerable pairs of PMUs, each measuring a single phasor. They also discuss how vulnerable pairs of PMUs can be targeted simultaneously in order to maximize an attack. Their techniques require that a specific attack-angle matrix is of rank equal to $1$, we refer to such attacks as rank-1 TSAs. 

The authors of~\cite{SDBDBP19} give a similar technique to find and undetectably attack a set of at least two vulnerable PMUs, each measuring a single phasor. They also show that when targeting at least three PMUs, the solution set of undetectable-attack offsets forms a continuum. They use this property to overcome constraints posed by clock controllers that prevents too large offsets. Their new attack strategies are to target at least three PMUs with gradually increasing offsets; this, in time, maximizes the attacker's objective. 

In~\cite{SDBDBP19} a sufficient and necessary vulnerability condition is identified for sets of PMUs with different time references, each measuring a single phasor. Their theory shows that rank-1 TSAs are feasible on a set of PMUs if and only if the index of separations (IoS) of the attack-angle matrices computed for each pair of PMUs are equal to $1$. The IoS is explained in Section~\ref{subsec:und}. This condition depends on the measurement values and enables them to find all vulnerable sets of PMUs that measure a single phasor. They also identified a sufficient vulnerability condition that does not depend on the measurements. They showed that if the infimum of the IoS, over all possible measurement values, is equal to $1$ for all pairs of PMUs in the targeted set, then rank-1 TSAs are always feasible.

The attack strategies and vulnerability conditions proposed in~\cite{rank1} and~\cite{SDBDBP19} require that each offset alters a single PMU that measures a single phasor. Yet, PMUs are capable of measuring several phasors, and several PMUs can share the same time-reference. In order to understand and mitigate rank-1 TSAs on real grids, it is important to generalize the theory established in~\cite{SDBDBP19} to sets of PMUs that possibly share the same time reference and that measure an arbitrary number of phasors. Also, the techniques in~\cite{SDBDBP19} all rely on a complex verification matrix corresponding to the least squares (LS) residuals when in fact the most commonly used state estimation technique is the weighted least squares (WLS) estimator. For exact attackability, the conditions are equivalent with the LS and the WLS matrices. However, for the establishment of vulnerability metrics, the use of the WLS matrices is more accurate because the WLS and LS residuals can be different. Due to the fact that the noise of phasors does not have circular symmetry, we cannot use the complex system model of~\cite{SDBDBP19} to compute the WLS estimates, we introduce the system model in rectangular coordinates. \md{In this paper, we establish the exact attackability conditions of grids with sets of phasors measured by PMUs that share the same time reference. For these exact conditions, we use the system model in complex form, i.e. the LS verification matrix, because it leads to a structural vulnerability condition that does not depend on the measurement values. Such a structural condition cannot be derived using the WLS matrices. We then use the system model in rectangular coordinates in order to provide vulnerability metrics that are linked to the WLS residuals.}

We group PMUs, each measuring one or more phasors, in sites if they share the same time-reference. Hence, an offset on the time-reference of a site impacts all of the corresponding phasor measurements in the same manner. Our first contribution is to identify a sufficient and necessary vulnerability condition for rank-1 TSAs that target a single site that measures an arbitrary number of phasors: item (b) of Theorem~\ref{th:single}. We also provide a metric that reflects how vulnerable a site is at a given time.

If two sites are not vulnerable by themselves, they could still be vulnerable as a pair. Our second contribution is to identify vulnerability conditions to rank-1 TSAs, for such pairs of sites. One is a sufficient condition that does not depend on the measurements: item (b) of Theorem~\ref{th:double}. The other is a sufficient and necessary general condition that depends on the measurement values: item (a) of Theorem~\ref{th:double}. We provide a second metric that reflects how vulnerable a pair of sites is at a given time.
 
 Our third contribution is to show that rank-1 TSAs, on a set of more than two sites, are feasible if and only if they are feasible for any pair of sites within the set. 
 
Finally, our fourth contribution is to mitigate the feasibility of rank-1 TSAs on a grid by combining our first three results. We establish a security requirement to prevent vulnerabilities identified by item (b) of Theorem~\ref{th:double}. Hence, to prevent vulnerabilities that do not depend on the measurements. We provide a greedy recursive algorithm that enforces our requirement, it takes as input the measurement points of an observable grid and it outputs a larger set of measurement points corresponding to the input set and additional points. A grid satisfying our requirement can still be vulnerable if it satisfies item (b) of Theorem~\ref{th:single} or item (a) of Theorem~\ref{th:double}. Although we conjecture that this is unlikely to occur, we still recommend the monitoring of the two provided metrics.

The paper is structured as follows. Section~\ref{sec:notation} defines the system model, the attack model and explains the undetectability conditions. Sections~\ref{sec:single} and ~\ref{sec:double} give the exact undetectability conditions for single sites and pairs of sites, respectively. They also provide vulnerability metrics, discuss the feasibility of the vulnerability conditions and compare our new conditions with the ones of previous papers. Section~\ref{sec:arbitrary} proves our third contribution. Section~\ref{sec:att} provides security measures to mitigate the feasibility of rank-1 TSAs on any observable grid. An algorithm to secure grids against structural vulnerabilities is also provided in Section~\ref{sec:att}. Numerical results, on the IEEE-39 bus benchmark with real load profiles from the Lausanne grid, are given in Section~\ref{sec:sim}. They show that the measurements of a grid satisfying our security requirements are far from satisfying the identified vulnerability conditions. Finally, Section~\ref{sec:ccl} concludes the paper.

\begin{table}
\caption{Notation used in this paper}
\begin{tabularx}{\columnwidth}{@{}XX@{}}
\toprule
  Indices \\
  $m$ & Number of PMU measurements \\
  $n$ & Number of buses \\
  $q$ & Number of attacked time references \\
  $p$ & Number of attacked measurements \\
  Parameters \\
  $z= (z_{\Box}^{[1,m]}+jz_{\Box}^{[m+1,2m]})^T \in \mathbb{C}^m$   & Complex measurement vector  \\
  $z_{\Box} =(\text{Re}(z), \text{Im}(z))^T\in \mathbb{R}^{2m}$   & Measurement vector in rectangular coordinates \\
  $x= (x_{\Box}^{[1,n]}+jx_{\Box}^{[n+1,2n]})^T \in \mathbb{C}^n$   & Complex state vector  \\
  $x_{\Box}=(\text{Re}(x), \text{Im}(x))^T \in \mathbb{R}^{2n}$   & State vector in rectangular coordinates \\
  $H \in \mathbb{C}^{m \times n}$ & Complex measurement-to-state matrix \\
  $H_{\Box} \in \mathbb{R}^{2m \times 2n}$ & Measurement-to-state matrix in rectangular coordinates \\
  $e \in \mathbb{C}^m$ & Complex error vector \\
  $e_{\Box} \in \mathbb{R}^{2m}$ & Error vector in rectangular coordinates \\
  $\hat{x} \in \mathbb{C}^n$ & Estimated complex state vector \\
  $\widehat{x_{\Box}} \in \mathbb{R}^{2n}$ & Estimated state vector in rectangular coordinates\\
  $\hat{z} \in \mathbb{C}^m$ & Estimated complex measurement vector \\
  $F \in \mathbb{C}^{m \times m}$ & Complex verification matrix to compute the complex LS residuals\\
  $Id$ & Identity matrix \\
  $S_i$ & Set of phasor indices at site $i$ \\
  $F^{S_i} \in \mathbb{C}^{m \times |S_i|}$ & Submatrix of $F$ with column indices in $S_i$ \\
  $z^{S_i} \in \mathbb{C}^{|S_i|}$ & Complex measurements with indices in $S_i$ \\
 
  $C_{\Box} \in \mathbb{R}^{2m \times 2m}$ & Covariance of the measurement noise in rectangular coordinates\\
  $r(z_{\Box}) \in \mathbb{R}^{2m}$ & Residual vector in rectangular coordinates \\ 
  $G_{\Box}=\begin{bmatrix} G_{\Box, 1} & G_{\Box, 2} \\ 
G_{\Box, 3} & G_{\Box, 4} \end{bmatrix} \in \mathbb{R}^{2m \times 2m}$ & Verification matrix to compute the normalized WLS residuals in rectangular coordinates\\
  $G_{\Box, 1},~G_{\Box, 2},~G_{\Box, 3},~G_{\Box, 4} \in \mathbb{R}^{m \times m}$ & Blocks of $G_{\Box}$ \\
$R = \frac{1}{2} \begin{bmatrix} G_{\Box, 1} -j G_{\Box, 2} \\ 
G_{\Box, 3} -j G_{\Box, 4} 
\end{bmatrix}$ & Matrix useful to compute $r(z_{\Box})=Rz + \bar{R}\bar{z}$  \\
  $\alpha$ & Attack angle in rad \\
  $z^a$, $z_{\Box}^a$ & Attacked measurement vector \\
  $\Delta z$, $ \Delta z_{\Box}$ & Difference between the attacked and unattacked measurement vectors \\
  $u_i= e^{j \alpha_i}$ & Attack value \\
  $u_i^*$ & Specific computed $u_i$ value \\
  $\varphi \in \mathbb{N}^{m \times q}$ & Attack indicator matrix \\
  $W \in \mathbb{C}^{q \times q}$ & LS attack angle matrix\\
  $(n_i)_i $ and $(v_i)_i \in \mathbb{C}^{p}$ & Vectors that span the null space of $F^{[S_1,...,S_q]}$ \\
  $I^i$, $V^i$ and $S_{inj}^i \in \mathbb{C}^*$ & Complex current, voltage and injected power values at time instant $i$ \\
   $l \in \mathbb{C}^*$ & Colinearity coefficient between two vectors \\
   $A \in \mathbb{C}^{q \times q}$ & WLS rank-1 approximation attack angle matrix \\
   Metrics \\
   $IoS_{i,t}(z_i,z_t) \in [0,1]$ & Index of Separation at sites $i$ and $t$ with respect to measurement values $z_i$ and $z_j$ \\
   $IoS_{i,t}^* \in [0,1]$ & Infimum of $IoS_{i,t}(z_i,z_t)$ over all possible values of $(z_i,z_t)$\\
   $ERR(F^{[S_1,..,S_q]})\in [0,1]$ & Effective rank ratio of $F^{[S_1,..,S_q]}$ \\
   
\bottomrule
\end{tabularx}
\end{table}
\section{System and Attack Models}\label{sec:notation}

We now introduce the complex system model that is the same as in~\cite{rank1,SDBDBP19} and the system model in rectangular coordinates. We then present the attack model and define TSAs and rank-1 TSAs. 
 
 \subsection{System Model in Complex Form}

We consider a grid consisting of $n$ buses. The state estimation is based on $m$ PMU measurements. Therefore, they correspond to electrical phasors such as currents or voltages. We suppose that the PMUs are placed at various sites of the grid. Each site is equipped with an arbitrary number of PMUs, each simultaneously measuring an arbitrary number of phasors. By 'arbitrary number', we mean $0$ or more according to the choices of the grid engineers. We suppose that the PMUs within a single site share the same time reference and that the PMUs in different sites have a different time reference. Note that by time reference, we refer to a site's clock and not to the time source, which could be space or network based. A same time source is used to synchronize the different time references (i.e. clocks) throughout the grid.

Because PMUs measure directly voltage and current phasors, the measurement vector $z\in \mathbb{C}^m$ and the voltage state vector $x \in \mathbb{C}^n$ are linearly linked by the following equation 
$$z=Hx+e,$$ where $e\in \mathbb{C}^m$ is the measurement error and the topology matrix $H \in \mathbb{C}^{m\times n}$ is composed of admittance-matrix and binary values. We say that the system is observable when $H$ is full rank and we say that a set of measurements is critical when removing it renders the system non-observable, i.e., the sub matrix of $H$ obtained by removing the corresponding rows is not full rank.

The LS estimate of the state vector has a closed-form solution $\hat{x}=(H^\dagger  H)^{-1}H^\dagger  z,$ where $H^\dagger$ denotes the complex conjugate transpose of $H$. From the estimated LS state vector $\hat{x}$, we can construct the LS estimated measurement vector $\hat{z}=H\hat{x}$ and compute its difference with the observed measurement vector $z$. The result of this difference is the LS residual vector. Define the verification matrix $F=H(H^\dagger  H)^{-1}H^\dagger - Id$, where $Id$ refers to the identity matrix, then $Fz=0$ if and only if there exists a state vector $x$ such that $z=Hx$. In other words, $Fz$ is the residual vector after the LS estimation. 
For ease of explanation, we introduce $S_i$ as the set of phasor indices in site $i$, $F^{S_i}$ as the $m \times |S_i|$ submatrix of $F$ corresponding to the columns of $F$ with indices in $S_i$ and $z^{S_i}$ as the vector of measurements with indices in $S_i$.  

The following two theorems facilitate the understanding of the link between the criticality of a set of measurements with the rank of the verification matrix $F$, their proofs are given in appendix.
\begin{theorem}\label{cl:main}
 For a set $S$ of measurement indices, $F^S$ is full rank if and only if measurements with indices in $S$ forms a non-critical set.
\end{theorem}

\begin{theorem}\label{th:p-k}
For a set $S$ of measurement indices, if rank$(F^S)=|S|-s$ with $0\leq s \leq |S|$, then all subsets of size $|S|-s+1$ is critical and there is at least one set of size $|S|-s$ that is non critical.
\end{theorem}

Theorem~\ref{th:p-k} implies that if a set of measurements with indices in $S$ is such that rank$(F^S)=1$, then all pairs of measurement within the set are critical and that at least one measurement is non critical.

\subsection{System Model in Rectangular Coordinates}\label{subsec:rec}

In previous papers~\cite{rank1},~\cite{SDBDBP19} undetectable TSAs are proposed such that the LS residuals are unchanged. The formulation of the system model in complex form facilitates the derivation of the theory. However, in practice, it is the WLS estimator that is used for state estimation. With the WLS estimator, the measurements are weighted by their noise variance. In a phasor measurement, noise of different standard deviation can appear in the phase and/or the magnitude. According to~\cite{rank1}, if the noise of the polar coordinates is Gaussian and small, then the noise of the rectangular coordinates can also be assumed to be small and Gaussian. However, we cannot assume that the noise of the real and imaginary parts of a phasor measurement are homoscedastic. Therefore, we cannot assume circular symmetry in the measurement errors which means that we cannot easily express the WLS estimates using complex matrix operations. To overcome this, we use matrices and vectors in rectangular coordinates. The measurement-to-state equation now becomes $$z_{\Box}=H_{\Box}x_{\Box}+e_{\Box}, $$ where $z_{\Box} = (\text{Re}(z), \text{Im}(z))^T \in \mathbb{R}^{2m}$ is the measurement vector in rectangular coordinates, $x_{\Box}\in \mathbb{R}^{2n}$ is the state vector in rectangular coordinates, $H_{\Box} \in \mathbb{R}^{2m \times 2n}$ is the rectangular-measurements-to-rectangular-state matrix and $e_{\Box} \in \mathbb{R}^{2m}$ is the error vector in rectangular coordinates. The corresponding WLS state estimate is $$\widehat{x_{\Box}}=(H_{\Box}^\dagger C_{\Box}^{-1} H_{\Box})^{-1}H_{\Box}^{\dagger}C_{\Box}^{-1}z_{\Box},$$ where $C_{\Box} \in \mathbb{R}^{2m \times 2m}$ is the covariance matrix of the measurement noise in rectangular coordinates. The computation of $C_{\Box}$, from both measurement values and measurement noise standard deviation, is described in~\cite{arpan}. 
We express the rectangular coordinates of the WLS measurement residuals as a complex relation, using the complex conjugate of $z$ as a variable:  
\begin{align*}
r(z_{\Box}) &=\begin{pmatrix} r_{re} \\ 
r_{im} \end{pmatrix} = G_{\Box} z_{\Box}=
\begin{bmatrix} G_{\Box, 1} & G_{\Box, 2} \\ 
G_{\Box, 3} & G_{\Box, 4} \end{bmatrix}
\begin{pmatrix} \text{Re}(z) \\ 
\text{Im}(z) \end{pmatrix} \\
&= 
\frac{1}{2} \begin{bmatrix} G_{\Box, 1} -j G_{\Box, 2} & G_{\Box, 1} +j G_{\Box, 2} \\ 
G_{\Box, 3} -j G_{\Box, 4} & G_{\Box, 3} +j G_{\Box, 4}\end{bmatrix}
\begin{pmatrix} z \\
 \bar{z} \end{pmatrix} = R z + \bar{R} \bar{z} ,\\
\end{align*}
where $\bar{z}$ is the complex conjugate of $z$, $G_{\Box, 1},~G_{\Box, 2},~G_{\Box, 3},~G_{\Box, 4} \in \mathbb{R}^{m \times m}$ are blocks of the $2m \times 2m$ real verification matrix $G_{\Box}=H_{\Box} (H_{\Box}^\dagger C_{\Box}^{-1} H_{\Box})^{-1}H_{\Box}^{\dagger}C_{\Box}^{-1}-Id$ and where the complex matrix $R \in \mathbb{C}^{2m \times m}$ is computed from blocks of the real matrix $G_{\Box}$: 
$$R = \frac{1}{2} \begin{bmatrix} G_{\Box, 1} -j G_{\Box, 2} \\ 
G_{\Box, 3} -j G_{\Box, 4} 
\end{bmatrix}.$$ 
As shown in~\cite{abur2004power} the WLS residuals are distributed according to a Gaussian distribution centered in $0$ and of covariance matrix $\Omega= (Id-H_{\Box} (H_{\Box}^\dagger C_{\Box}^{-1} H_{\Box})^{-1}H_{\Box}^{\dagger}C_{\Box}^{-1})C_{\Box}$. 

Note that the verification matrix with respect to the LS estimation $F$ does not depend on measurement values but only on the topology of the grid. In comparison, the verification matrix with respect to the WLS estimation $G_{\Box}$, depends on the covariance of the measurement noise, which itself depends on the measurement values~\cite{arpan}.

\subsection{Attack Model}

We suppose that an attacker can observe the measurement vector $z$, that he knows the topology of the system (i.e. he knows $H$) and that he is able to manipulate the time reference of $q$ sites via a GPS spoofing attack or a delay box insertion. The attack on $q$ sites affects a total of $p \geq q$ measurements. An injected time offset $d$ in the time reference of a site directly shifts the phase of all measured phasors in this site by an angle of $\alpha =2 \pi f d$ rad, where $f \approx 50$ or $60$ Hz is the instantaneous voltage frequency and the offset $d$ is in seconds. Therefore, a measurement $z_i$ affected by an attack angle $\alpha$ is of the form $z_i^a=z_ie^{j \alpha}$: its phase is shifted by the attack angle and its magnitude is unchanged. Note that the injection of an offset in the time reference of a site shifts the phase of all phasors measured in this site by the same attack angle. Such an attack is a TSA.

\subsection{Undetectability Condition}\label{subsec:und}

An attack is said to be undetectable if it is not flagged by the BDD techniques in place. The majority of such implemented techniques are residual based, specifically, they analyze the difference between the observed and estimated measurements~\cite{LA18,GA15,AB19,7839276}. To the best of our knowledge, residual-based BDD techniques are mostly variants of the largest normalized residual (LNR) test and the $\chi^2$-test. The first flags abnormally large LNR values and the second flags an abnormal distribution of the sum of the residuals.
Hence, an attack carefully crafted to ensure that the residuals are unchanged can have an impact on the state estimation and not be detected by residual-based BDD techniques. Thus, we say that an attack is undetected if and only if $r(\Delta z_{\Box})= G_{\Box}  \Delta z_{\Box}=0$, where $\Delta z_{\Box} = z_{\Box}^a-z_{\Box}$ is the difference between the attacked and unattacked measurement vectors. The following lemma \md{shows that for exact attackability, the use of $F$ and $R$ is equivalent.}

\begin{theorem}\label{lem:eq}
For a measurement vector $z$, $$r(z_{\Box})=0 \iff Fz=0 \iff Rz=0 .$$
\end{theorem}
\begin{proof}
$r(z_{\Box})=0 \iff \exists x_{\Box}=(x_{\Box,1},x_{\Box,2})^T \in \mathbb{R}^{2n} \text{ s.t } x_{\Box,1},x_{\Box,2} \in \mathbb{R}^n \text{ and } z_{\Box} = H_{\Box} x_{\Box}$ $\iff  z=Hx \text{, with }  x=(x_{\Box,1}+jx_{\Box,2})^T \in \mathbb{C}^n \iff Fz=0$.

If $r(z_{\Box})=0$, then $R z + \bar{R} \bar{z}=0$, hence $R z =- \bar{R} \bar{z}$. Furthermore, if $r(z_{\Box})=0$ then there exist a state vector $x \in \mathbb{C}^n$ such that $z=Hx$. Hence, for any $\lambda \in \mathbb{C}$, $\lambda z=H(\lambda x)$ and thus $r(\lambda z_{\Box})=0$. Taking $\lambda =j$, we get $r(z_{\Box})= jR z - j\bar{R} \bar{z}=0$, hence $R z = \bar{R} \bar{z}$. Because $R z =- \bar{R} \bar{z}= \bar{R} \bar{z}$, we conclude that $Rz=0$.

If $Rz=0$, then $r(z_{\Box})=R z + \bar{R} \bar{z}=0$.
\end{proof}
\md{We choose to use $F$ to establish exact vulnerability conditions because $R$ depends on the measurement values and $F$ does not. This enables us to provide a structural vulnerability condition that we cannot derive from $R$. In contrast, we use $R$ for vulnerability metrics because it is linked to the WLS residuals. Thus it enables us to measure how detectable an attack is, or how vulnerable to attacks a set of sites is.}

Suppose that the attacker introduces $q$ time offsets to $q$ sites, following the notation in~\cite{rank1}, define the attack vector $$u=(u_1, \cdots ,u_q)^T=(e^{j \alpha_1}, \cdots, e^{j \alpha_q})^T\in \mathbb{C}^q,$$ and the $m \times q$ indicator matrix $\varphi$ such that $$\varphi_{t,i}=  \begin{cases} 1 &\text{ if measurement of index }t\text{ is affected by angle } \alpha_i, \\ 0 &\text{ otherwise.} \end{cases} $$ It was proven in~\cite{rank1} that the attacks are absolutely undetectable if and only if $\sum_{i=1}^q(u_i -1)F \text{diag}(z) \varphi_{:,i}=0,$ where $\varphi_{:,i}$ is the $i^{th}$ column of $\varphi$ and $\diag (z)$ is the diagonal matrix with measurement values along the diagonal. For scalability, the authors of~\cite{rank1} introduced the $q \times q$ attack-angle complex hermitian matrix 
$W=\varphi^T\text{diag}(z)^\dagger F^\dagger F \text{diag}(z) \varphi.$
They showed that an attack vector $u$ is undetectable if and only if $W(u-1)=0$. 
To the best of our knowledge, all known techniques to compute undetectable attack offsets require that the rank of $W$ is equal to $1$. We refer to such attacks as rank-1 TSAs.

\jylb{If for a pair of sites, that each measures a single phasor, the rank of $W$ is not equal to $1$, the authors of~\cite{rank1} and~\cite{SDBDBP19} show that it is sometimes possible to use a rank-1 approximation of $W$ for rank-1 TSAs. This is possible when the index of separation (IoS), defined in~\cite{rank1} as the largest eigenvalue of $W$ over the sum of both its eigenvalues, is close to $1$. This occurs if the largest eigenvalue is significantly larger than the remaining one. In~\cite{rank1}, the infimum of the IoS over all measurement values is introduced as IoS$^*$. This quantity does not depend on the measurements. If it is equal to or close to $1$, then the IoS is always equal to or close to $1$. This condition is used, in~\cite{rank1} and~\cite{SDBDBP19}, to find vulnerable sets of PMUs, that each measures a single phasor, from the verification matrix only. For sites measuring an arbitrary number of phasors we are required, in Section~\ref{sec:double}, to study the effective rank of rectangular matrices with possibly more than two singular values. For this purpose, we introduce the effective rank ratio (ERR) of a matrix as its largest singular value over the sum of all its singular values. The ERR is close to $1$ if the largest singular value is significantly larger than the others, in which case the effective rank of the matrix is close to $1$. We use this condition to find vulnerable sets of sites that each measure an arbitrary number of phasors.}

\section{Vulnerability Conditions for a Single site}\label{sec:single}

We now provide a novel sufficient and necessary condition for rank-1 TSAs that target a single site. We then propose a vulnerability metric to measure how vulnerable a site is. Finally, we discuss the feasibility of the identified condition. 

\subsection{Vulnerability Condition}

We assume that an attacker injects an offset in the time reference of $q=1$ site; this affects $p$ measurements. The remaining measurements are not affected by the attack. Note that an attack on a site with all $p$ measurements equal to zero does not have any effect because the attacks create only phase shifts and a phase shift on $0$ is still equal to $0$. Hence, we study the vulnerability of sites that measure at least one non-zero synchrophasor. 
The following theorem gives necessary conditions in order to mount a rank-1 TSA.

\begin{theorem}\label{th:single}
Consider a rank-1 TSA on a single site measuring $p \geq 1$ phasors with indices in $S_1$, such that no measurement alone is critical and at least one measurement is not equal to $0$. 
\begin{enumerate}
\item[(a)] If $p=1$: such an attack is never feasible.
\item[(b)] If $p \geq 2$: such an attack is feasible if and only if $z^{S_1}$ is in the null space of $F^{S_1}$.
\end{enumerate}
\end{theorem}

\begin{proof}
When $q=1$, notice that by definition $W$ corresponds to a single complex value:
$W= \left( \sum_{i \in S_1} F_{:,i} z_i \right)^{\dagger}\left( \sum_{i \in S_1} F_{:,i} z_i \right).$ In this case, a rank-1 TSA corresponds to a single complex value $u$ such that $W(u-1)=0$. As mentioned in~\cite{rank1}, non-trivial attacks exist in this case if and only if $W=0$. Thus, if and only if  \begin{equation}\label{eq:rel}
\sum_{i \in S_1} F_{:,i} z_i =0.
\end{equation}

\begin{enumerate}
\item[(a)] If $p=1$, then Eq.~(\ref{eq:rel}) is equivalent to $F_{:,1} z_1 =0.$ In other words, because $z_1$ is non-zero, a rank-1 TSA on a single site that measures a single phasor is feasible if and only if its corresponding column of $F$ is equal to the null vector. Hence, such an attack is feasible if and only if the phasor is critical. As we assume that no single measurement is critical, we conclude that if  $p=1$, then no rank-1 TSA is feasible. 

\item[(b)] If $p \geq 2$, observe that the left-hand-side of Eq.~(\ref{eq:rel}) is equal to $F^{S_1}z^{S_1}$. Therefore, Eq.~(\ref{eq:rel}) is satisfied if and only if the targeted measurement vector is in the null space of $F^{S_1}$. 
\end{enumerate}
\end{proof}
Observe that if a site is vulnerable to rank-1 TSAs, then any attack angle will be undetected.

\subsection{Vulnerability Metric: distance to item (b) of theorem~\ref{th:single}}\label{subsec:dist1}

If a site is not vulnerable to rank-1 TSAs, it might still be close to it. In other words, if the residuals obtained after an attack are different but very close to the residuals obtained without an attack, the attack is undetectable in practice. The following theorem shows that $\| R^{S_1}z^{S_1} \|$ is a reliable vulnerability metric for single sites. The closer this metric is to $0$, the more vulnerable the site is.

\begin{theorem}
$\| r( \Delta z_{\Box}) \| = \text{ cst}\| R^{S_1}z^{S_1} \|,~ 0 \leq \text{cst} \leq 4 ~ \forall u_1 \in \mathbb{C}$ such that $|u_1|=1$.
\end{theorem}
 \begin{proof}
  Values of the complex vector $R^{S_1}z_{\Box}^{S_1}$ can be written in polar form $\rho_i e^{j \phi_i}$ and any $u_1 \in \mathbb{C}$ such that $|u_1|=1$ can also be written in polar form $u_1=e^{j \alpha}$. Hence, 
  \begin{align*}
  \| r( \Delta z_{\Box})  \| &=\| R \Delta z + \bar{R} \bar{ \Delta z}\| = \| 2 \text{Re}(R \Delta z) \| = \| 2 \text{Re}((u_1 -1)R^{S_1}z^{S_1})\| \\
   &= \sqrt{4\sum_{i=1}^{m} \rho_i^2 (cos(\phi_i + \theta) - cos(\phi_i))^2 } \\
   &= \text{ cst}\sqrt{ \sum_{i=1}^{m} \rho_i^2} = \text{ cst} \| R^{S_1}z^{S_1} \|, \text{ with } 0\leq \text{cst} \leq 4 .
  \end{align*}
 \end{proof}

\subsection{Feasibility of the Vulnerability Condition: item (b) of theorem~\ref{th:single}}\label{sec:singledisc}

If a rank-1 TSA is feasible, then there exists a relation among the measurement values with coefficients computed from values of vectors spanning $\ker (F^{S_1})$. For example, in the case where $p=2$, an attack is feasible if and only if the two involved measurements $(z_1,z_2)$ are such that they satisfy the relation 
\begin{equation}\label{eq:rels1p2}
\frac{z_1}{z_2}=\frac{n_1}{n_2},
\end{equation} 
where $(n_1,n_2)^T$ is a fixed vector that spans $\ker (F^{S_1})$. This is because $F^{S_1}$ is an $m$ by $2$ matrix, hence its rank is at most equal to $2$. As none of the measurements are critical by themselves, its rank cannot be equal to $0$ and as $z^{S_1}$ is in the null space of $F^{S_1}$, its rank cannot be equal to $2$. Hence, its rank is equal to $1$. By the rank theorem, $\ker (F^{S_1})$ is of dimension equal to $1$. Hence, $z^{S_1}$ must be a non-zero complex multiple of any vector that spans $\ker (F^{S_1})$, which leads to Eq.~(\ref{eq:rels1p2}). 

Intuitively, such a relation seems unlikely to occur as measurement values depend on independent loads. For example, if $z_1$ is a voltage value $V$ and $z_2$ a current value $I$, then at all time instant $i$, a rank-1 TSA is feasible if and only if $$\frac{z_1^i}{z_2^i}=\frac{V^i}{I^i}=\frac{V^i I^{\dagger i}}{| I^i|^2} = \frac{S_{inj}^i}{| I|^2}= \frac{n_1}{n_2}.$$ Hence, at all time instants the phase of the complex injected power $S_{inj}$ must remain constant and equal to $\arg{\frac{n_1}{n_2}}.$ Also, observe that the relation implies that $\frac{|V^i|}{|I^i|}= \frac{|n_1|}{|n_2|}.$ As the magnitude of voltage values is always close to $1$, the magnitude of the current values is approximately invariant and close to $\frac{|n_2|}{|n_1|}.$ In practice, the injected current and power depends on the loads that vary in time due to external factors. Hence, it seems unlikely that they could be of constant magnitude or phase and even less likely that such constant values could be equal to specific values computed from the verification matrix only. Therefore, we conclude that the necessary conditions for a rank-1 TSA on a single site that measures two phasors are unlikely to occur on a realistic grid.

We present numerical values of the metric $\| R^{S_1}z^{S_1} \|$ obtained through realistic simulations in Section~\ref{sec:sim}. In our simulations, we always encountered large values far from $0$, which corroborates our intuition that necessary conditions to mount a rank-1 TSA on a single site seem unrealistic; but we still recommend to monitor it on a real grid.
\section{Vulnerability Conditions for a pair of sites}\label{sec:double}

If two sites are not vulnerable by themselves to rank-1 TSAs, they could still be attackable together. We identify two novel conditions to simultaneously mount rank-1 TSAs on such two sites. One of them is a sufficient vulnerability condition that is structural as it depends on the LS verification matrix only. The other is a general necessary and sufficient condition that depends on the measurement values. We then provide a vulnerability metric to measure the vulnerability of any pair of sites. We also discuss the feasibility of the identified general vulnerability condition. Finally, we establish the relation between our novel conditions and the vulnerability conditions presented in~\cite{rank1,SDBDBP19} for the special case where both sites measure only a single phasor (i.e. $p=2$).

\subsection{Vulnerability Conditions}

We suppose that an attacker injects different offsets in the time reference of $q=2$ sites, hence affecting a total of $p$ phasor measurements with indices listed in $S_1$ and $S_2$ for the first and second sites, respectively. The following theorem gives necessary conditions in order to mount a rank-1 TSA.

\begin{theorem}\label{th:double}
Consider a rank-1 TSA on $q=2$ sites measuring phasors with indices in $S_1$ and $S_2$, respectively, such that $|S_1|+|S_2|=p$, no measurement is critical by itself, neither site is vulnerable to rank-1 TSAs by itself and at least one measurement in each site is not equal to zero.
\begin{enumerate}
\item[(a)] Such an attack is feasible if and only if $F^{S_1}z^{S_1}$ and $F^{S_2}z^{S_2}$ are colinear.
\item[(b)] If $\rank (F^{[S_1,S_2]})=1$, i.e. if all pairs of measurements in $S_1 \cup S_2$ are critical: such an attack is always feasible.
\item[(c)] If $\rank (F^{[S_1,S_2]})=p$, i.e. if the combined set of measurements $S_1 \cup S_2$ is not critical: such an attack is never feasible.
\end{enumerate}
\end{theorem}

\begin{proof}
Rank-1 TSAs are feasible if and only if $\rank (W)=1$. By the rank properties of complex matrices, we have that $\rank (W) =  \rank ((F  \mbox{diag}(z) \varphi)^{\dagger} (F  \mbox{diag}(z) \varphi)) = \rank (F  \mbox{diag}(z) \varphi).$ Hence, rank-1 TSAs are feasible if and only if 
\begin{equation}\label{eq:rankf}
 \rank  \left(  \begin{bmatrix}  \sum_{i \in S_1} F_{1,i} z_i & \sum_{i \in S_2} F_{1,i} z_i \\ \vdots & \vdots \end{bmatrix} \right) =1.
\end{equation}
\begin{enumerate}
\item[(a)] Eq.~(\ref{eq:rankf}) is satisfied if and only if either
\begin{itemize}
\item one of the columns of the matrix is equal to the null vector. According to Theorem~\ref{th:single} this is equivalent to saying that a rank-1 TSA can be mounted directly on the corresponding site. However, this case is excluded because we suppose that no rank-1 TSA can be mounted on a site by itself.
\item or the two columns of the matrix are colinear. Specifically $F^{S_1}z^{S_1}$ and $F^{S_2}z^{S_2}$ are colinear. 
\end{itemize}
\end{enumerate}

Since $F^{S_1}z^{S_1}$ and $F^{S_2}z^{S_2}$ are non-zero, they are colinear if and only if there exists an $l \in \mathbb{C}^*$ such that $F^{S_1}z^{S_1}=lF^{S_2}z^{S_2} \iff F^{[S_1,S_2]}(z^{S_1},lz^{S_2})^T=0 \iff  (z^{S_1},lz^{S_2})^T \in \ker (F^{[S_1,S_2]})$
$\iff \dim (E_Z \cap E_N)  \geq 1,$ with $E_Z=  \text{span} \{ (z^{S_1}, 0)^T, (0, z^{S_2})^T \}$ and $E_N= \ker (F^{[S_1,S_2]}).$
According to Grassmann's formula, this is equivalent to
$\dim (E_Z) + \dim (E_N) - \rank ([Z|N]) \geq 1,$
where $[Z|N]= \begin{bmatrix} z^{S_1} & 0 & N \\ 0 & z^{S_2} & \vdots \end{bmatrix}$ and where $N$ is a matrix with independent columns that span $\ker (F^{[S_1,S_2]}).$ Therefore, rank-1 TSAs are feasible if and only if
\begin{align}
2 + p- \rank (F^{[S_1,S_2]})-1 & \geq \rank ([Z|N]) \nonumber\\
\iff 1+p- \rank (F^{[S_1,S_2]}) & \geq \rank ([Z|N]). \label{eq:intdim}
\end{align}

\begin{enumerate}
\item[(b)] If $\rank (F^{[S_1,S_2]})=1$, then Eq.~(\ref{eq:intdim}) $ \iff p \geq \rank ([Z|N]).$ Observe that in this case, $[Z|N]$ is a $p$ by $p+1$ matrix. Hence, it is always the case that the rank of $[Z|N]$ is at most equal to $p$. In other words, a rank-1 TSA is always feasible. 

\item[(c)] If $\rank (F^{[S_1,S_2]})=p$, then Eq.~(\ref{eq:intdim}) $\iff 1 \geq \rank (Z),$ which is never satisfied as $Z$ is a matrix of rank equal to $2$. In other words, a rank-1 TSA is never feasible.
\end{enumerate}
\end{proof}

\jylb{Interestingly, Theorem~\ref{th:double} implies that if three different sites $S_i$, $S_j$ and $S_k$, that are not vulnerable by themselves, are such that two pairs $(S_i,S_j)$ and $(S_i,S_k)$ are vulnerable, then $(S_j,S_k)$ is also vulnerable. In other words, item $(a)$ of theorem~\ref{th:double} defines an equivalence relation over the set of sites that are not vulnerable by themselves.}

As mentioned previously, Theorem~\ref{th:double} establishes two novel vulnerability conditions:
\begin{itemize}
\item a \emph{general vulnerability} condition defined by item (a) of Theorem~\ref{th:double}: this condition is necessary and sufficient for vulnerability to rank-1 TSAs. It depends on the measurement values.
\item a \emph{structural vulnerability} condition defined by item (b) of Theorem~\ref{th:double}: this sufficient vulnerability condition depends on the LS verification matrix only. In other words, if two sites satisfy this condition, then they are vulnerable but it is also possible that two sites that do not satisfy this condition are vulnerable.
\end{itemize}
If a pair of sites is not exactly \emph{structurally} vulnerable to rank-1 TSAs, it can be practically so. Specifically, if the rank of $F^{[S_1,S_2]}$ is not equal to $1$, its effective rank can be close to $1$. This distance is captured by the ERR of $F^{[S_1,S_2]}$. Numerically, it can occur that a column of $F^{[S_1,S_2]}$ has significantly larger values than the remaining columns, in which case the effective rank is smaller than the computed rank. As a result, ERR$(F^{[S_1,S_2]})$ can be very close to $1$. In this case it is likely that a rank-1 TSA computed from a rank-1 approximation of $F^{[S_1,S_2]}$, is feasible in practice irrespective of the measurement values.

We present numerical results through realistic simulations in Section~\ref{sec:sim}. Our simulations show that when the rank of $F^{[S_1,S_2]}$ is not equal to $1$, its ERR is sometimes very close to $1$, in which case the corresponding pair of sites is practically vulnerable. 

\subsection{General Vulnerability Metric: distance to item (a) of theorem~\ref{th:double}}\label{subsec:dist2}

As in the single-site analysis, when a pair of sites is not exactly vulnerable to rank-1 TSAs, it can still be close to vulnerable. Theorem~\ref{lem:eq} implies that item (a) of theorem~\ref{th:double} can be equivalently written using the general WLS notations introduced in Section~\ref{subsec:rec} as: a pair of sites is vulnerable if and only if $R^{S_1}z^{S_1}$ and $R^{S_2}z^{S_2}$ are colinear. The details of this proof are given in appendix~\ref{ap:c}. In other words, an undetectable attack requires that the rank of the $m \times 2$ complex matrix $\begin{bmatrix} R^{S_1}z^{S_1} | R^{S_2}z^{S_2} \end{bmatrix} $ is equal to $1$. If it is not the case, the metric of vulnerability that we propose is the ERR of $\begin{bmatrix} R^{S_1}z^{S_1} | R^{S_2}z^{S_2} \end{bmatrix} $. This metric shows how vulnerable a set of measurements corresponding to a pair of sites is. If the metric is equal to or approximately equal to $1$, then the pair of sites is vulnerable in practice.

\subsection{Feasibility of the General Vulnerability Condition: item (a) of theorem~\ref{th:double}}\label{sec:doubledisc}

If a pair of structurally non-vulnerable sites satisfy the general vulnerability condition, there exists a relation between the measurements with coefficients computed from the verification matrix. Specifically, if $\rank (F^{[S_1,S_2]}) \neq 1$ but $z^{S_1}$ and $z^{S_2}$ are such that $F^{S_1}z^{S_1}$ and $F^{S_2}z^{S_2}$ are colinear, then at least one of the measurements is directly determined by a combination of the other measurements and values of the vectors spanning $\ker (F^{[S_1,S_2]})$. 

For example if phasors $(z_1,z_2)$ and $(z_3,z_4)$ are measured at the first and second sites, respectively, then a rank-1 TSA on the two sites is feasible if and only if either
$$\frac{z_1}{z_2}=\frac{n_1}{n_2} \text{ and }\frac{z_3}{z_4}=\frac{n_3}{n_4},$$ where $(n_1,n_2,n_3,n_4)^T$ is a fixed vector that spans the null space of $F^{[S_1,S_2]}$;
or $$ \frac{z_3}{z_4}= \frac{z_1(v_3n_2-v_2n_3)+z_2 (v_1n_3-v_3n1)}{z_1(v_4n_2-v_2n_4)+z_2 (v_1n_4-v_4n1)},$$ where $ (n_1,n_2,n_3,n_4)^T$ and $(v_1,v_2,v_3,v_4)^T$ are independent vectors that span the null space of $F^{[S_1,S_2]}$. The proof of this is not given here but is similar to the one given in Section~\ref{sec:singledisc}. As for the single-site vulnerability condition discussed in Section~\ref{sec:singledisc}, we conjecture that such relations are unlikely to occur on real grids, as measurement values depend on independent loads. 

In our simulations presented in Section~\ref{sec:sim}, the observed values of ERR$(\begin{bmatrix} R^{S_1}z^{S_1} | R^{S_2}z^{S_2} \end{bmatrix}) $ for pairs of sites that do not satisfy the structural vulnerability condition are far from $1$. This corroborates our intuition that the general vulnerability condition seems unlikely to occur for pairs of sites that are not already structurally vulnerable to rank-1 TSAs. 

\subsection{Relation with the Vulnerability Conditions of~\cite{rank1,SDBDBP19}}

We now show that if each site measures a single phasor (i.e. $p=2$), then our vulnerability conditions are equivalent to the vulnerability conditions identified in~\cite{rank1,SDBDBP19}.
The following theorem gives the equivalence for the general vulnerability condition.

\begin{theorem}
Consider a rank-1 TSA on $q=2$ sites, each measuring a single phasor $z_1$ and $z_2$, respectively, such that no measurement is critical by itself or is equal to zero. Then $IoS_{1,2}(z_1,z_2)=1$ if and only if $F^{S_1}z_1$ and $F^{S_2}z_2$ are colinear.
\end{theorem}

\begin{proof}
As there are only $2$ involved measurements, $W$ is a $2$ by $2$ matrix. As no measurement is critical by itself, the rank of $W$ is equal to either $1$ or $2$.
By definition, the IoS of $W$ is equal to $1$ if and only if its smallest eigenvalue is equal to $0$, which is equivalent to $\rank (W)=1$. Recall from the proof of Theorem~\ref{th:double} that the rank of $W$ is equal to $1$ if and only if Eq.~(\ref{eq:rankf}) is satisfied. As no measurement is critical by itself, this is equivalent to the colinearity of $F^{S_1}z_1$ and $F^{S_2}z_2$.
\end{proof}

Similarly, the following theorem gives the equivalence for the structural vulnerability condition. 

\begin{theorem}\label{th:compstruct}
Consider a rank-1 TSA on $q=2$ sites, each measuring a single phasor $z_1$ and $z_2$, respectively, such that no measurement is critical by itself or is equal to $0$. Then  $IoS_{1,2}^*=1$ if and only if $\rank (F^{[S_1,S_2]})=1$. 
\end{theorem}

\begin{proof}
We show both directions:
\begin{itemize}
\item $IoS_{1,2}^*=1 \rightarrow \rank (F^{[S_1,S_2]})=1:$ By definition of $IoS^*$, if it is equal to $1$, then the $IoS$ is equal to $1$, whatever the values of $z_1$ and $z_2$. Therefore, Eq.~(\ref{eq:rankf}) is satisfied even if $z_1=z_2=1$, which is equivalent to $\rank (F^{[S_1,S_2]})=1$.
\item $IoS_{1,2}^*=1 \leftarrow \rank (F^{[S_1,S_2]})=1:$ It was shown in~\cite{rank1} that if $p=2$, then 
\begin{equation}\label{eq:star}
IoS^*_{1,2}=\frac{1}{2}+\frac{|f_{12}|}{2(f_{11}f_{22})^{1/2}},
\end{equation} with $f_{it}=\sum_{l,m} \sum_{n} \varphi_{l,i} \varphi_{m,t} \bar{F}_{n,l}F_{n,m}.$
Notice that \\
$|f_{12}|^2 = \left( \sum_{i=1}^m \bar{F}_{i,1}F_{i,2} \right) \left( \sum_{i=1}^m \bar{F}_{i,2}F_{i,1} \right),$\\ 
$~f_{11} = \sum_{i=1}^m \bar{F}_{i,1}F_{i,1}$ and $f_{22} = \sum_{i=1}^m \bar{F}_{i,2}F_{i,2}.$ \\
If $\rank{F^{[S_1,S_2]}}=1$, then there exists $l \in \mathbb{C}^*$ such that $F_{:,1}=lF_{:,2}$ because no measurement is neither equal to $0$ nor critical. Hence, $f_{11}=|l|^2f_{22} \text{ and }  |f_{12}|^2=|l|^2f_{22}^2.$
By plugging them into Eq.~(\ref{eq:star}), we get that $IoS^*=1.$
\end{itemize}
\vspace{-0.3cm}
\end{proof}

Note that $IoS_{1,2}(z_1,z_2)$ is the IoS of the attack-angle matrix $W$ computed from measurements $(z_1,z_2)$. Also, $IoS_{1,2}^*$ is the infimum of $IoS_{1,2}(z_1,z_2)$ over all possible values of $(z_1,z_2)$. Both $ERR(F^{[S_1,S_2]})$ and $IoS_{1,2}^*$ are independent of measurement values. Theorem~\ref{th:compstruct} implies that one of them is equal to $1$ if and only if the other is also equal to $1$. Therefore, they are both equal to $1$ for an exactly structurally vulnerable pair of sites measuring a single phasor. However, as we show next, $ERR(F^{[S_1,S_2]})$ may be close to $1$ when $IoS_{1,2}^*$ is not. In other words, our structural vulnerability condition is better than the one in~\cite{rank1} and~\cite{SDBDBP19} as it identifies more sets that are practically vulnerable. The following claim is proven in appendix.

\begin{claim}\label{cl:err}
For any pair of sites, $1 \geq ERR(F^{[S_1,S_2]}) \geq IoS_{1,2}^*$ and $ERR(F^{[S_1,S_2]})$ can be close to $1$ when $IoS_{1,2}^*$ is not if $F^{S_1}$ has much larger values than $F^{S_2}$ or vice-versa.
\end{claim}

\section{Vulnerability Conditions for an Arbitrary Number of sites}\label{sec:arbitrary}

We now show that a rank-1 TSA targeting an arbitrary number of sites $q \geq 2$ that are not vulnerable by themselves, is feasible if and only if for every pair of sites among the targeted set of sites, a rank-1 TSA is feasible.
As a result, the vulnerability analysis of a grid reduces to the vulnerability analysis of each site and each pair of sites. We then establish the relation between our result and those identified in~\cite{SDBDBP19} for the special case where all targeted sites measure only a single phasor.

\subsection{Vulnerability Conditions}

The following theorem establishes the general vulnerability condition for an arbitrary number of sites that are not vulnerable by themselves to rank-1 TSAs.
\begin{theorem}\label{th:arbitrary}
A set of $m \geq q \geq 2$ sites, each measuring an arbitrary number of phasors, such that none of the sites are vulnerable to rank-1 TSAs by themselves, are vulnerable together if and only if each pair of sites within the set is vulnerable to rank-1 TSAs.
\end{theorem}

\begin{proof}
Define $S_1$, $S_2$, ..., $S_q$ to be the set of measurement indices corresponding to phasors measured in the first, second, up to $q^{th}$ targeted sites, respectively.
Then, the rank of the attack-angle matrix $W$ corresponds to the rank of the following matrix
\begin{equation*}
T=
\begin{bmatrix}
\sum_{i \in S_1} F_{1,i}z_i & \cdots & \sum_{i \in S_q} F_{1,i}z_i\\
\vdots & &\vdots 
\end{bmatrix}.
\end{equation*} 
The rank of this $m$ by $q$ matrix is equal to $1$ if and only if all columns of $T$ are dependant. As none of the sites are vulnerable to rank-1 TSAs by themselves, no column of $T$ is equal to the null vector. Therefore, all sub matrices consisting of two columns of $T$ must be of rank equal to $1$. This means that the attack-angle matrix $W$ that corresponds to the attack targeting the corresponding two sites is of rank equal to $1$. In other words, if a set of $q$ sites is vulnerable to rank-1 TSAs, then any two sites within the set of targeted sites are also vulnerable to rank-1 TSAs.
 
Similarly, if there is a set of $q$ sites such that all pairs of sites within the set are attackable undetectably, then all columns of $T$ are dependant, which means that its rank is equal to $1$. Therefore the large set of $q$ sites is vulnerable to rank-1 TSAs.
\end{proof}

Theorem~\ref{th:arbitrary} implies that a site is vulnerable to rank-1 TSAs either if it is vulnerable by itself or if its combination with at least one other site forms a vulnerable set where all pairs of sites are vulnerable together. Hence, mitigating the feasibility of rank-1 TSAs for each site and each pair of sites is sufficient to mitigate the attack feasibility of the grid. Table~\ref{table:sumup} recapitulates the vulnerability conditions and distance to vulnerability metrics for a site and for a pair of sites with non-critical measurements.

\begin{table*}[t]
\centering
\begin{tabular}{l|l|l|l|}
\cline{2-4}
                                    & structural vulnerability                                               & general vulnerability & distance metric \\ \hline
\multicolumn{1}{|l|}{single site}   &           none                                                  & $F^{S_1} z_{S_1}=0$          & $\| R^{S_1} z_{S_1} \|\approx 0$   \\ \hline
\multicolumn{1}{|l|}{pair of sites} & ERR$(F^{[1,2]}) \approx 1$ & $F^{S_1}z^{S_1}$ and $F^{S_2}z^{S_2}$ colinear & ERR$(\begin{bmatrix} R^{S_1}z^{S_1} | R^{S_2}z^{S_2} \end{bmatrix})\approx 1$     \\  \hline
\end{tabular}
\caption{Vulnerability conditions and distance to vulnerability metric for each site and pair of sites of a grid with non-critical measurements.}
\label{table:sumup} 
\end{table*}

\subsection{Relation with the Results of~\cite{SDBDBP19}}

The authors of~\cite{SDBDBP19} show that measurements can be grouped in equivalence classes, when the IoS values of the attack-angle matrices of all pairs of measurements are equal to $1$. Then, they show that a set of measurements is vulnerable to rank-1 TSAs if and only if the set of measurements belong to the same equivalence class. In other words, they show that a set of sites, each measuring a single phasor, is vulnerable to rank-1 TSAs if and only if all pairs of sites within the targeted set is vulnerable. \jylb{Recall that item $(a)$ of theorem~\ref{th:double} defines an equivalence relation over the set of sites that are not vulnerable by themselves. Hence, Theorem~\ref{th:arbitrary} shows that a set of non-vulnerable sites is vulnerable if and only if the sites belong to the same equivalence class. Therefore,Theorem~\ref{th:arbitrary} applied to sites measuring a single phasor coincides with the result of~\cite{SDBDBP19}.}

\section{Mitigating rank-1 TSAs}\label{sec:att}

To minimize the feasibility of rank-1 TSAs, we now combine results from Sections~\ref{sec:single},~\ref{sec:double} and~\ref{sec:arbitrary}. We propose a greedy offline algorithm to ensure that no pair of sites is structurally vulnerable. Even if a grid is not structurally vulnerable, there is still an unlikely possibility that measurement values are such that some sites or pairs of sites are vulnerable. In order to check the non-vulnerability of the system, we recommend the monitoring of the vulnerability metrics.

\subsection{Securing against the Structural Vulnerability Condition}

Apart from vulnerability conditions, Theorem~\ref{th:double} also identifies a structural non-vulnerability condition. Item (c) of Theorem~\ref{th:double} states that if the combined set of measurements from the two sites does not form a critical set, then they are not vulnerable to rank-1 TSAs, irrespective of their measurement values. Hence, a natural idea to secure all pairs of sites is to ensure that none of them forms a critical set of measurements. However, from an engineering perspective, it is not realistic to impose this security measure, as it would either be impossible to enforce or require much redundancy in the measurements. This would require substantially more PMUs than what is required for observability of the system, which would be costly. For example, by placing PMUs on every bus of the grid used for simulations in Section~\ref{sec:sim}, we obtain that $34 \%$ of pairs are still critical.

In contrast, it is possible to carefully increase the measurement redundancy of an observable grid's PMU allocation such that no pair of sites are structurally vulnerable. Specifically, a security requirement is that the ERR of $F^{[S_1,S_2]}$ cannot be close to $1$ for all pairs of sites. 

Recall that theorem~\ref{th:p-k} implies that $ERR (F^{[S_1,S_2]})=1$ if and only if all pairs of measurements in $S_1 \cup S_2$ are critical. In an observable system, if a pair of sites is such that no measurement is critical and all pairs are critical, by adding one phasor measurement in one of the two sites, at least one pair of measurements will be non-critical. As a result, the ERR of $F^{[S_1,S_2]}$ will no longer be equal to $1$. However, the practical security requirement is stronger and requires that the ERR is not close to $1$. Therefore, to secure an observable but vulnerable grid, we propose to identify vulnerable pairs of sites such that the ERR of $F^{[S_1,S_2]}$ is close to $1$ and to iteratively increase the number of phasors that they should measure until no critical pair of sites has a high ERR$(F^{[S_1,S_2]})$ value. A greedy strategy is to increase the number of measured phasors at sites that appear most frequently in the list of vulnerable pairs of sites. \jylb{Algorithm~\ref{alg:sec} implements this greedy strategy by recursively building the secured set of measurement points.} It takes as input the set of measurement points of the observable grid and it outputs a larger set of measurement points which includes the input set and the additional phasors required for structural security. Note that this algorithm only secures against structural vulnerabilities, it can thus be performed offline at each change of topology. 

The authors of~\cite{SDBDBP19} present a rank-1 TSA targeting $q=5$ structurally vulnerable sites. In Section~\ref{sec:sim}, we secure the grid using Algorithm~\ref{alg:sec} and show that the attack is no longer feasible.

\begin{algorithm}[h]
  \caption{Secure-Grid($M$)}
  \begin{algorithmic}
  	\Require $M$ (set of measurement points for an observable grid)
  	\Ensure $M_{new}$ (updated set of measurement points for an observable and structurally non-vulnerable grid)
  	\texttt{\\}
  	\State $H \leftarrow$ Create topology matrix from admittance and $M$
  	\State $Vulnerable \leftarrow \emptyset$
  	\State $F \leftarrow H(H^\dagger  H)^{-1}H^\dagger - Id$
  	\For {site $i$ in grid}
  		\State $S_i=M(i) $
  		\For {$j \neq i$ in grid}
  			\State $S_j=M(j) $
  			\State $\Lambda \leftarrow \text{Singular-Values} (F^{[S_i,S_j]})$
  			\If {$\frac{\max(\Lambda)}{\sum \Lambda} \geq \eta$}  
  				\State $Vulnerable \leftarrow Vulnerable \cup {(i,j)}$
  			\EndIf
  		\EndFor
  	\EndFor
  	\If {$Vulnerable \neq \emptyset$}
  	\While {$Vulnerable\neq \emptyset$}
  		\State $freq \leftarrow$ Get the most frequent index in $Vulnerable$

  			\State $M \leftarrow M \cup freq$
  			\For {$tuple \in Vulnerable$}
  				\If {$freq \in tuple$}
  					\State $Vulnerable \leftarrow$ Remove $tuple$ from $Vulnerable$
  				\EndIf
  	  			\EndFor	
  	\EndWhile
  	\State $M \leftarrow  \text{Secure-Grid}(M)$
  	  	\EndIf
  	  	\State $M_{new} \leftarrow M$ \\
  		\Return $M_{new}$
  \end{algorithmic}
   \label{alg:sec}
\end{algorithm}

\subsection{Monitoring the General Vulnerability Metrics}

Once the grid is secured against structural vulnerabilities, the measurements of sites or of pairs of sites can still satisfy the general vulnerability conditions identified by item (b) of Theorem~\ref{th:single} and by item (a) of Theorem~\ref{th:double}. As discussed in Sections~\ref{sec:singledisc} and~\ref{sec:doubledisc}, we conjecture that such conditions are unlikely to be satisfied in reality. By precaution, we propose to compute, at every estimation of the system's state, $\| R^{S_1} z^{S_1} \|$ for every site and ERR$(\begin{bmatrix} R^{S_1}z^{S_1} | R^{S_2}z^{S_2} \end{bmatrix}) $ for every pair of sites. If over time it can be observed that a site or a pair of sites is frequently close to vulnerability, then the measurements satisfy either exactly or approximately a dangerous relation. In this case, we recommend breaking this relation by modifying the PMU allocation around the corresponding site or pair of sites.


\section{Simulations}\label{sec:sim}

\jylb{We validate our results with simulations on the IEEE 39 bus benchmark with real load profiles taken from the Lausanne grid at $50$ Hz. We apply the security requirements, established in the previous section, and show that they achieve the desired security.}

\subsection{Electrical Model}
The PMU allocation we consider is depicted in Figure~\ref{gridunlikely}. Specifically, there are $12$ zero-injection buses, PMUs measuring both voltages and currents, are placed at buses $\{ 30, 37, 28, 38, 18, 39, 12, 16, 7, 31, 32, 34, 33,20,25,26,29 \}$ and PMUs measuring only currents are placed at buses $\{ 24,35,15,21,4,23,36 \}$. We define sites to be groups of buses separated by transformers only: $\{ 2,30\}$, $\{ 6,31\}$, $\{ 10, 32\}$, $\{ 11,12,13\}$, $\{19,20,33,34\}$, $\{22,35\}$,$\{ 23,36\}$, $\{ 25,37\}$, $\{ 29,38\}$, the other buses correspond to one-bus sites. With this allocation, no measurement is critical by itself, no PMU is critical by itself but some sites are critical. 

\begin{figure}
\centering
    \includegraphics[scale=0.45]{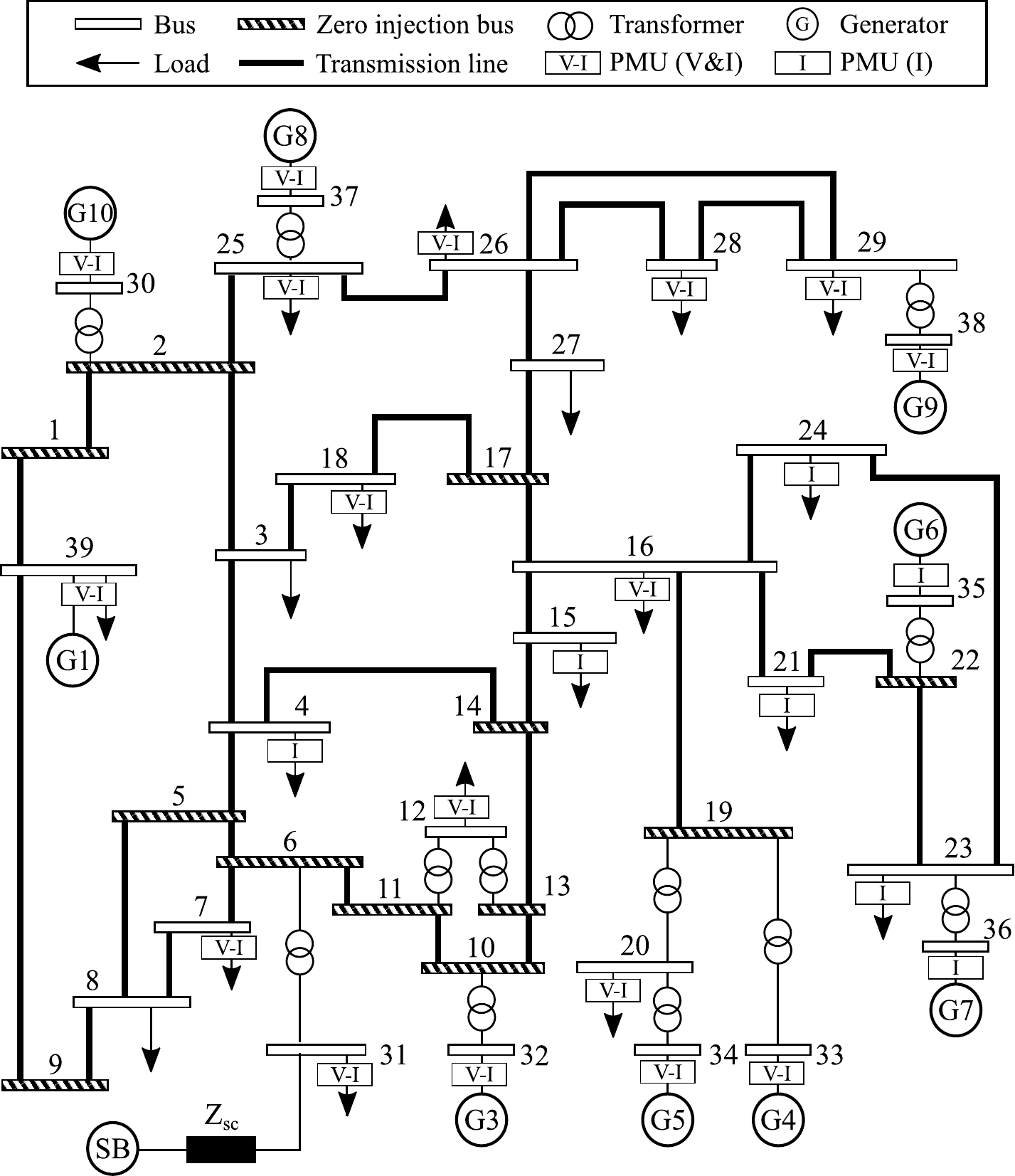}
\caption{ PMU allocation on the IEEE-39 bus benchmark. Buses separated by a transformer are grouped in a site: PMUs in a site share the same time reference.}
\label{gridunlikely}
\end{figure}

Our simulations are done every $20$ms over $700$s, thus at $35'000$ different time instants. At each time instant, we create a measurement vector by computing the load flow. This results in the true state of the system. We then add randomly generated Gaussian noise to the true state, which results in the simulated measurement vector $z$.

\subsection{Securing a Grid against Structural Vulnerabilities}\label{subsec:securing}

The grid described above features pairs of sites such that all pairs of measurements are critical, i.e. such that $\rank(F^{[S_1,S_2]})=1$. In other words, they are structurally vulnerable to rank-1 TSAs. It is the case for all pairs of sites among $\{ 21 \}$, $\{ 22,35 \}$, $\{ 23,36 \}$ and $\{ 24 \}$, which means that a rank-1 TSA can be mounted on any combination of these sites. In fact, the authors of~\cite{SDBDBP19} present rank-1 TSAs on various combinations of the involved buses. Therefore, the grid needs to be secured from this structural vulnerability. 

In order to gain insights on how vulnerable the other pairs of sites are in practice, we compute the ERR of their corresponding matrix $F^{[S_1,S_2]}$. The results given in Figure~\ref{fig:presec} reflect how close to $1$ the ERR of $F^{[S_1,S_2]}$ matrices are. We observe that the median of ERRvalues is at $0.556$ and that all values are between $0.2114$ and $0.9939$. The closer the ERR of $F^{[S_1,S_2]}$ is to $1$, the more vulnerable the pair of sites is. Clearly, we can observe that several pairs of sites are vulnerable to rank-1 TSAs in practice, even though they have a non-critical pair of measurements.
 
\begin{figure}
\centering
    \includegraphics[scale=0.5]{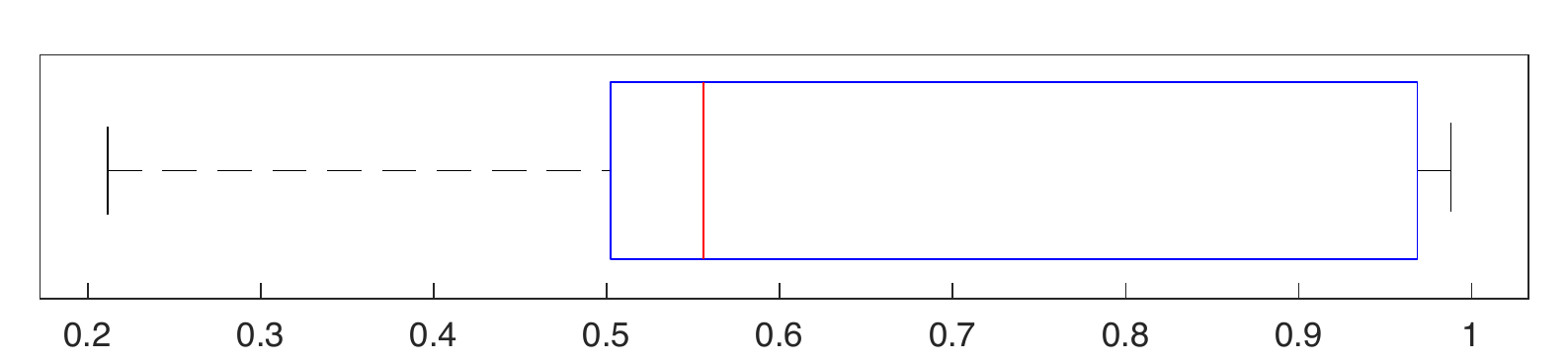}
\caption{Distribution of ERR values of $F^{[S_1,S_2]}$ matrices for all pairs of sites before applying Secure-Grid: some values are close to $1$, the corresponding pairs of sites are structurally vulnerable in practice. }
\label{fig:presec}
\end{figure}

Secure-Grid applied to this grid outputs a set of measurement points that includes additional phasor measurements. Specifically, buses $4$, $15$, $21$, $24$, $35$ and $36$ are required to measure an additional phasor. 
We observe on Figure~\ref{fig:sec} that there are no longer any pair of sites such that the ERR of $F^{[S_1,S_2]}$ is close to $1$. Also on Figure~\ref{fig:sec}, we observe that the IoS of the attack-angle matrix $W$, at all $35'000$ time instants is always much smaller than $1$; i.e., there are no grid locations that are structurally vulnerable to rank-1 TSAs. 

\begin{figure}
\centering
    \includegraphics[scale=0.5]{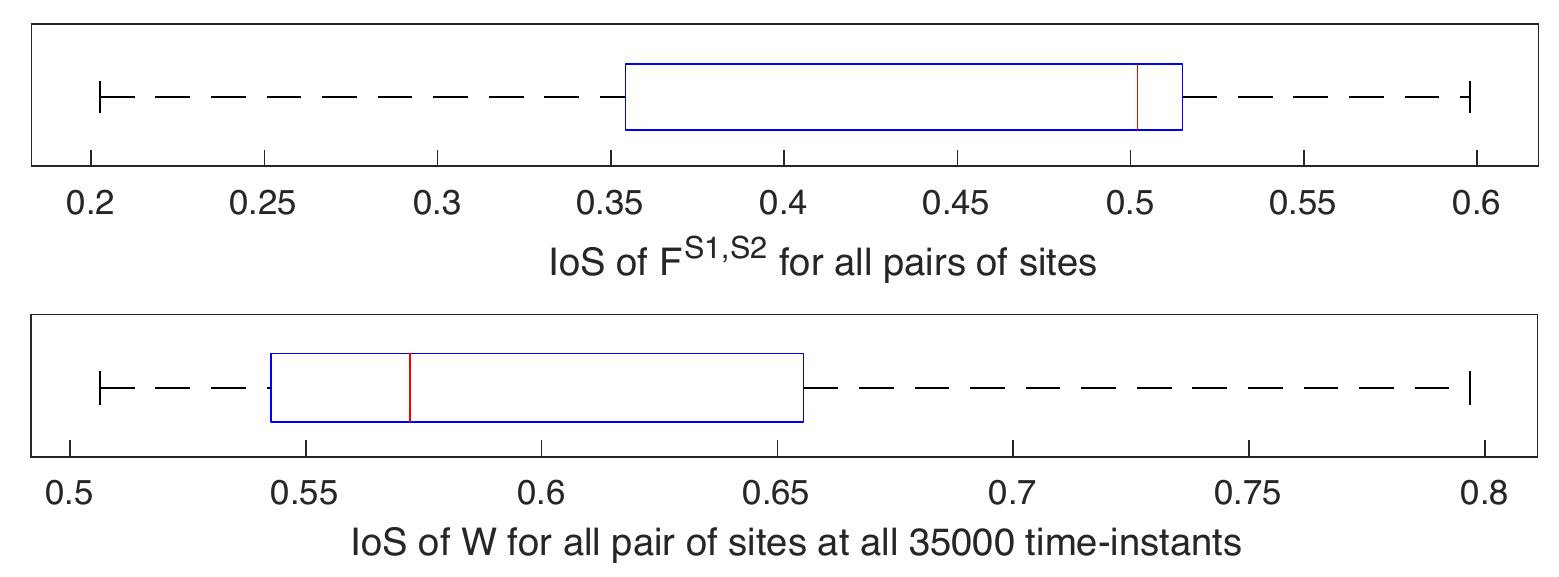}
\caption{Distribution of ERR values of $F^{[S_1,S_2]}$ and IoS values of $W$ matrices for all pairs of sites after applying Secure-Grid: none are close to $1$, no pairs of sites are structurally vulnerable in practice.}
\label{fig:sec}
\end{figure}

The authors of~\cite{SDBDBP19} also present an attack on a set of three buses with different time references, which involves a total of $5$ measurements. This attack is performed on a grid with a slightly different PMU allocation. Specifically, only one phasor is measured at bus $26$ and no phasors are measured at bus $29$. In this setting, they presented an attack on buses $26$, $28$ and $38$.  Notice that buses $28$ and $38$ both measure two phasors simultaneously. This successful attack was found by trials and errors but was not explained by the theory of~\cite{SDBDBP19}. We now understand that this set of buses was in fact structurally vulnerable because the rank of the corresponding sub matrix of $F$ is equal to $1$. We are now also able to secure the grid against this attack by adding phasors to measure at buses $26$ and $29$. Figure~\ref{fig:previousattack} compares the servo-aware attack impact and LNR values before using Secure-Grid (i.e. the results presented in~\cite{SDBDBP19}) and after using Secure-Grid. It shows that the undetectable attack from~\cite{SDBDBP19} becomes clearly detectable once additional phasors are measured at the buses identified by Secure-Grid. The change of values in the middle of the simulation is due to a sudden increase of factor $2$ in the active power introduced at a nearby bus. All attacks presented in~\cite{rank1} and~\cite{SDBDBP19} targeted sets of PMUs which were in fact structurally vulnerable to rank-1 TSAs. Our new security requirement therefore prevents all of them.

\begin{figure*}[t]
\centering
\includegraphics[width=\linewidth]{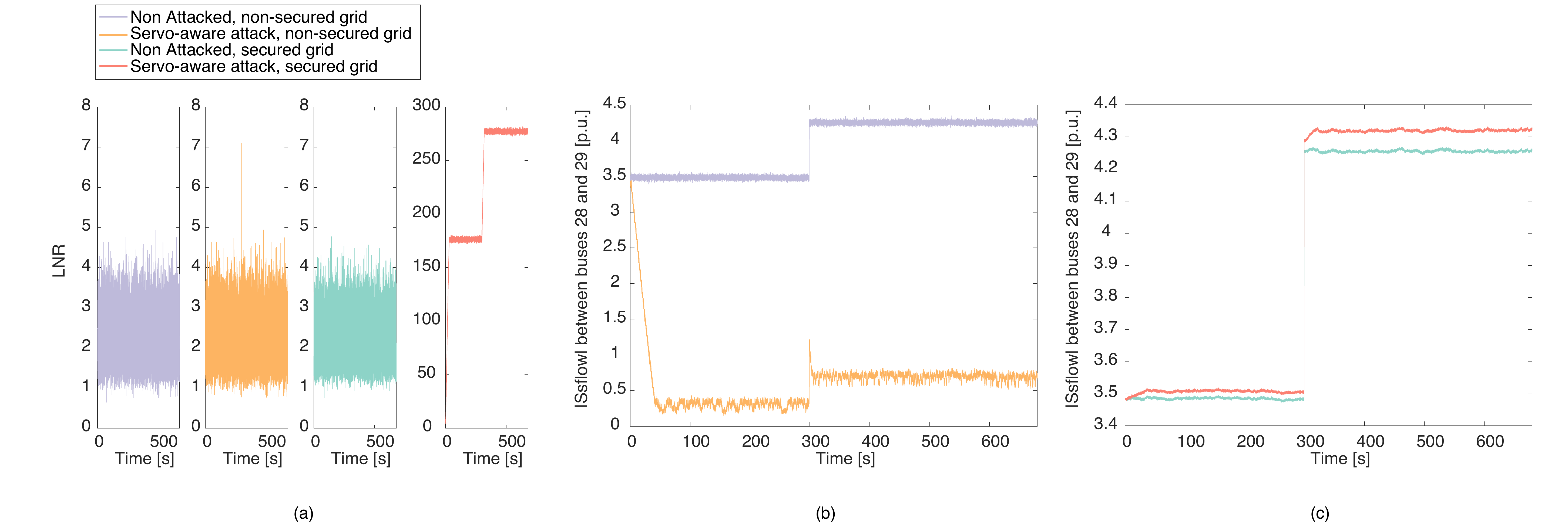}
\caption{Comparison of the impact and detectability of the servo-aware attack targeting buses $26$, $28$ and $38$ before and after using Secure-Grid: (a) LNR values show that the servo-aware attack is undetectable before Secure-Grid and detectable after it; (b) The magnitude of the power-flow with and without the attack \emph{before} Secure-Grid, the undetectable attack has a large impact; (c) The magnitude of the power-flow with and without the attack \emph{after} Secure-Grid, the undetectable attack has an impact but is very detectable.}
\label{fig:previousattack}
\end{figure*}

Even though the resulting grid is not structurally vulnerable, it is still possible that measurement values satisfy the general vulnerability conditions for a site or a pair of sites. We show that such specific conditions are far from appearing on our realistic grid.

\subsection{General Vulnerability Condition for Single Sites}\label{subsec:cond1}

At each of the $35'000$ time instants of the simulation, we compute the metric $\| R^{S_1}z^{S_1} \|$ introduced in Section~\ref{subsec:dist1}. Recall that this metric is equal to $0$ if and only if the site is vulnerable. The distribution of the obtained values of all sites are shown in Figure~\ref{fig:single}. We observe that the metric is never equal to $0$, the most vulnerable site has a metric that is on average equal to $0.664$. To illustrate that $0.664$ reflects that the site is far from vulnerable in practice, we perform an attack on the corresponding site $\{ 31\}$ with a constant offset of $20 \mu s$, which is the maximum offset allowed by the PMU clock controllers in~\cite{SDBDBP19}. Recall that a vulnerable single site can be attacked undetectably with any attack-angle, including $20 \mu s$. The obtained largest normalized residuals are approximately $6$ times larger than those obtained without an attack. Such a difference is easily identified. In other words, the site that is the closest to satisfying the single-site vulnerability condition is far from vulnerable in practice. As a result, we observe that the measurement values of all sites of the grid are always far from satisfying the conditions necessary to mount a rank-1 TSA.
\begin{figure}
\centering
    \includegraphics[scale=0.45]{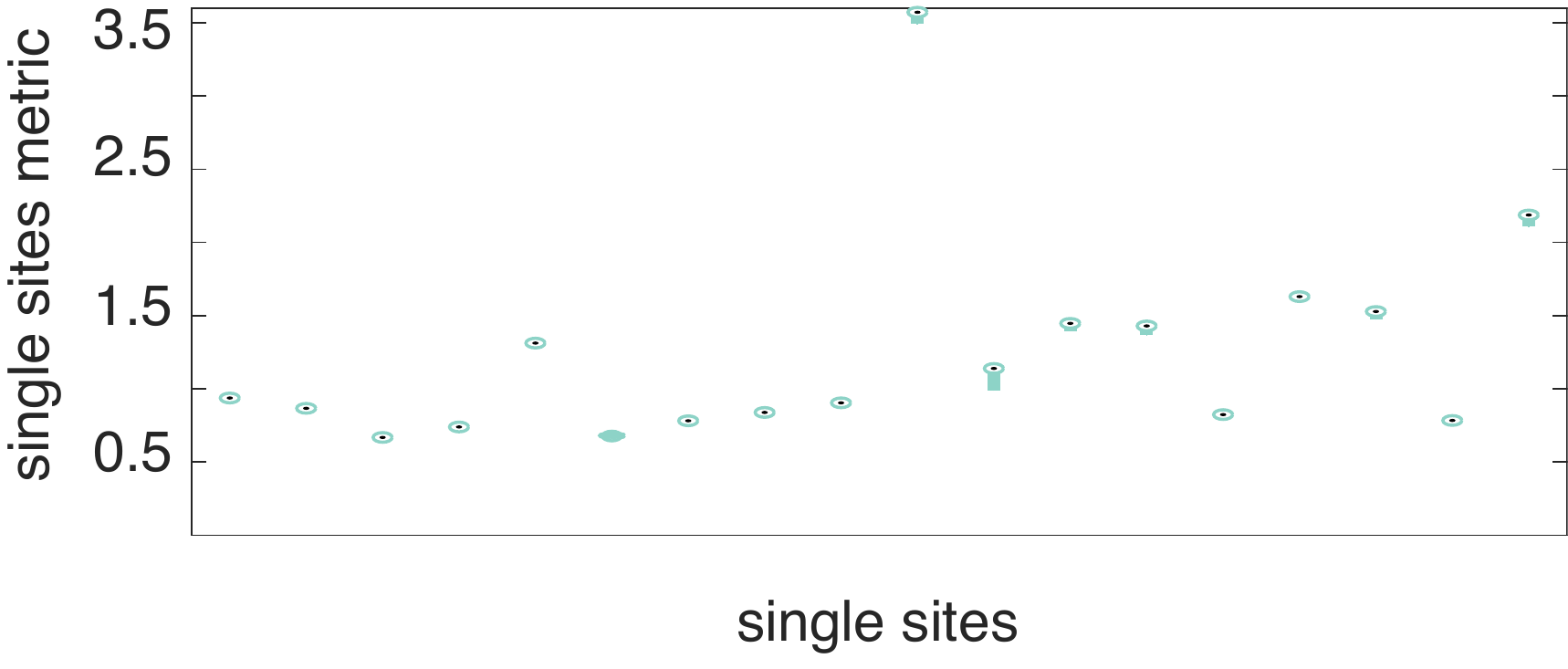}
\caption{After Secure-Grid: Distribution of the values of $\| R^{S_1}z^{S_1} \|$ for each site at $35'000$ time instants. It is never equal to $0$: the condition for a rank-1 TSA on a single site is never satisfied.}
\label{fig:single}
\end{figure}

\subsection{General Vulnerability Condition for Pairs of Sites}\label{subsec:cond2}

At each of the $35'000$ time instants of the simulation we compute the vulnerability metric introduced in Section~\ref{subsec:dist2} for each pair of sites. Specifically, for each pair of sites we compute ERR$(\begin{bmatrix} R^{S_1}z^{S_1} | R^{S_2}z^{S_2} \end{bmatrix}) $. Recall that the closer the metric is to $1$, the more vulnerable the pair of sites is. The distribution of the obtained values for every pair of sites at all time instants are given in Figure~\ref{fig:double}. We observe that the maximum value is $0.843$, which is not close enough to $1$ for undetectable attacks. Figure~\ref{pairdetect} shows that an attack on the corresponding pair of sites (i.e. the one with the highest vulnerability metric values) is detectable. All the other pairs have vulnerability metric values that are even further from satisfying the vulnerability conditions. As a result, we observe that at each time instant of the simulation, the measurement values of all pairs of sites of the secured grid are far from satisfying the conditions necessary to mount rank-1 TSAs.

\begin{figure}
\centering
    \includegraphics[scale=0.45]{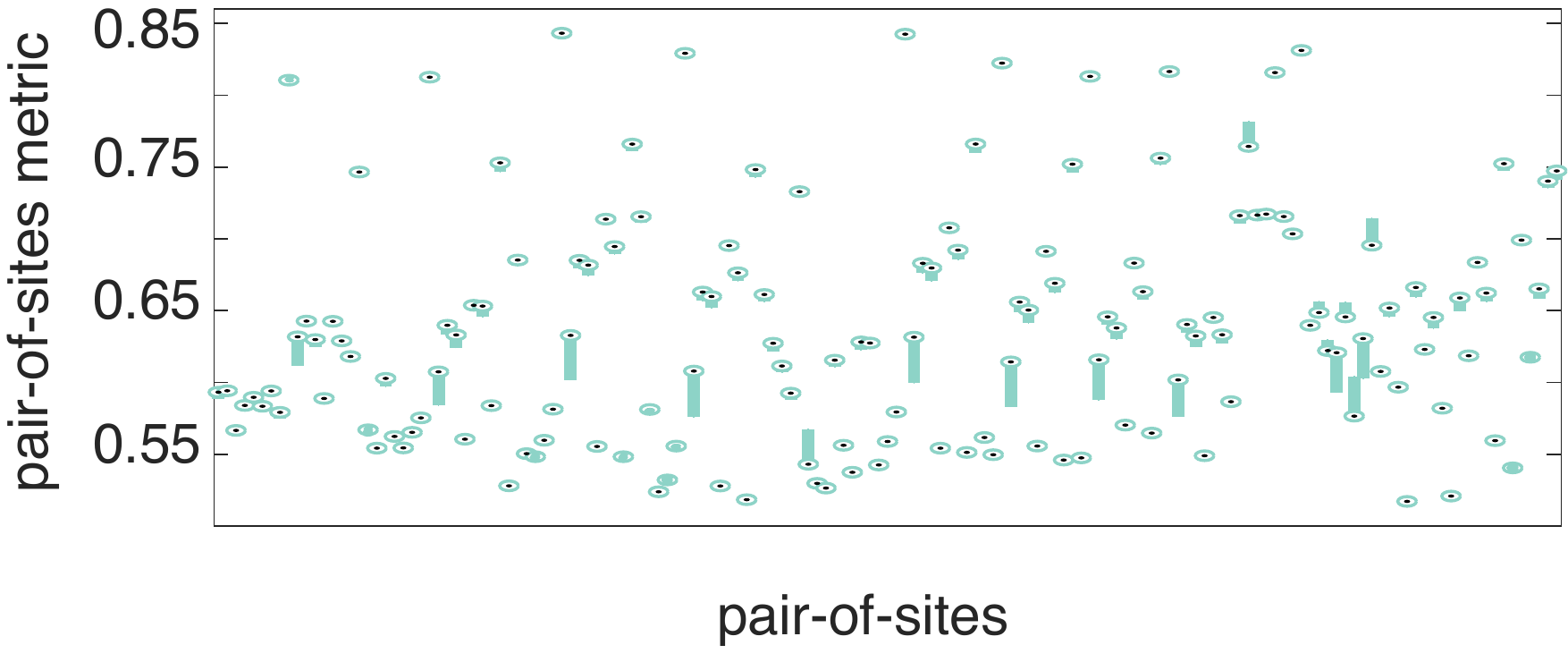}
\caption[]{After Secure-Grid: Distribution of the values of ERR$(\begin{bmatrix} R^{S_1}z^{S_1} | R^{S_2}z^{S_2} \end{bmatrix}) $ for each pair of sites at $35'000$ time instants. The most vulnerable pair has a mean metric value of $0.8439$, which is too far from $1$ for undetectable attacks.}
\label{fig:double}
\end{figure}

\begin{figure}
\centering
    \includegraphics[scale=0.45]{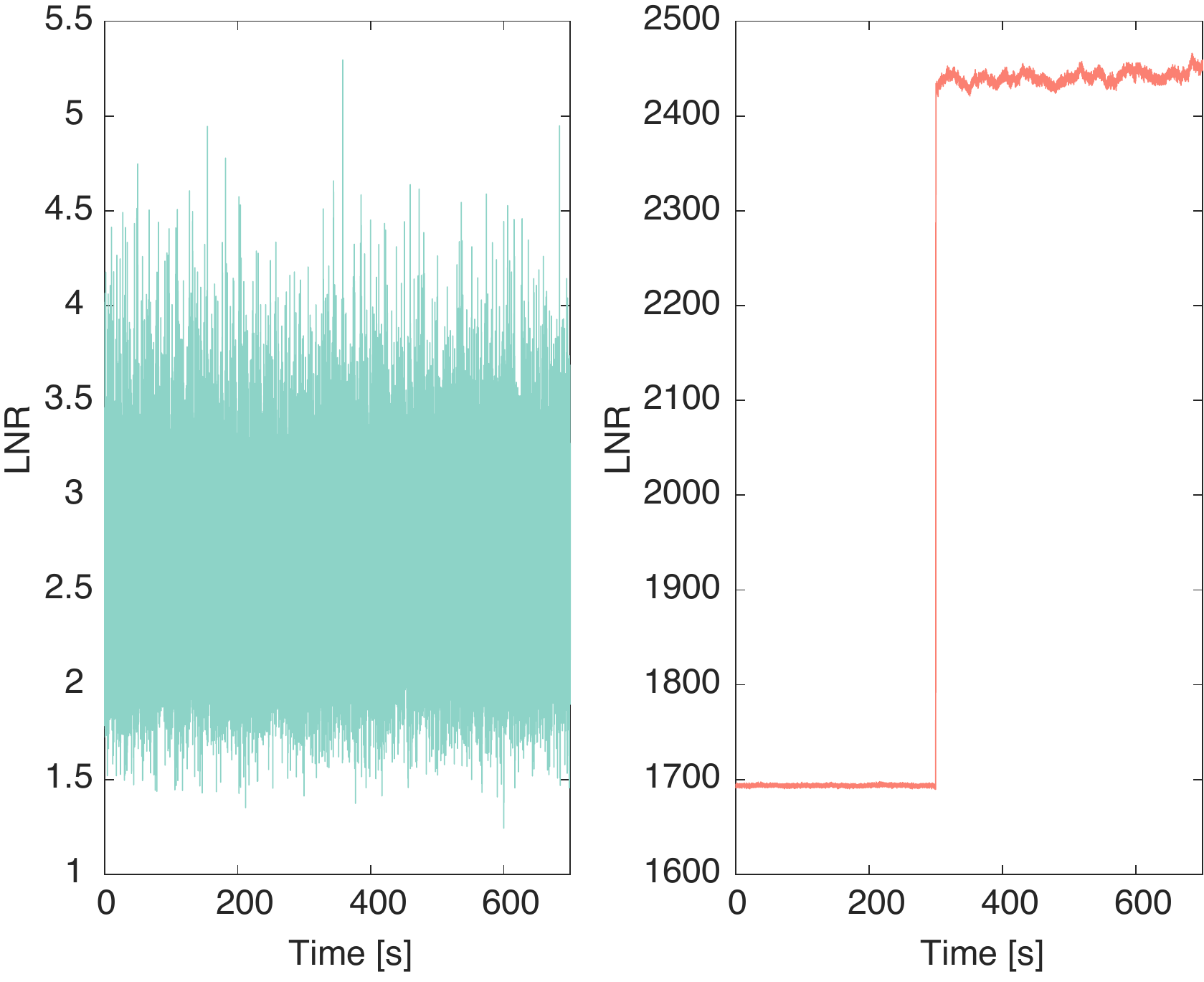}
\caption[]{After Secure-Grid: LNR values without an attack (left) and with an attack on the most vulnerable pair of sites (right); the values are clearly distinguishable, the attack is detected.}
\label{pairdetect}
\end{figure}

\section{Conclusion}\label{sec:ccl}

We showed that the analysis of the vulnerability of a grid to rank-1 TSAs reduces to the vulnerability analysis for every site and every pair of sites. We identified a sufficient vulnerability condition for pairs of sites that measure an arbitrary number of phasors. This condition does not depend on the measurement values. We established a security requirement to prevent this vulnerability. We also provided an offline greedy algorithm that enforces our security requirement. If our security requirement is satisfied, it is still possible that measurement values are such that attacks are feasible, although we conjecture that it is unlikely. We identified sufficient and necessary vulnerability conditions for single sites and for pairs of sites, each measuring an arbitrary number of phasors. We recommend the monitoring of two metrics associated to these conditions in order to check the non-vulnerability of the grid. Numerical results, on the IEEE-39 bus benchmark with real load profiles from the Lausanne grid, show that the measurements of a grid satisfying our security requirement are far from vulnerable to rank-1 TSAs. Finally, our results, applied to sites measuring a single phasor, coincide with the results of~\cite{SDBDBP19}. Our new theory aslo enables us to better understand and thus prevent attacks presented in~\cite{SDBDBP19}.

\bibliographystyle{IEEEtran}
\bibliography{references}

\begin{thebibliography}{10}
\providecommand{\url}[1]{#1}
\csname url@samestyle\endcsname
\providecommand{\newblock}{\relax}
\providecommand{\bibinfo}[2]{#2}
\providecommand{\BIBentrySTDinterwordspacing}{\spaceskip=0pt\relax}
\providecommand{\BIBentryALTinterwordstretchfactor}{4}
\providecommand{\BIBentryALTinterwordspacing}{\spaceskip=\fontdimen2\font plus
\BIBentryALTinterwordstretchfactor\fontdimen3\font minus
  \fontdimen4\font\relax}
\providecommand{\BIBforeignlanguage}[2]{{%
\expandafter\ifx\csname l@#1\endcsname\relax
\typeout{** WARNING: IEEEtran.bst: No hyphenation pattern has been}%
\typeout{** loaded for the language `#1'. Using the pattern for}%
\typeout{** the default language instead.}%
\else
\language=\csname l@#1\endcsname
\fi
#2}}
\providecommand{\BIBdecl}{\relax}
\BIBdecl

\bibitem{RP14}
P.~Romano and M.~Paolone, ``Enhanced interpolated-dft for synchrophasor
  estimation in fpgas: Theory, implementation, and validation of a pmu
  prototype,'' \emph{IEEE Trans. on Instrumentation and Measurement}, 2014.

\bibitem{TSG13}
Y.~hua Tang ; Gerard N. Stenbakken ; Allen~Goldstein, ``Calibration of phasor
  measurement unit at nist,'' \emph{IEEE Trans. on Instrumentation and
  Measurement}, 2013.

\bibitem{BFMP15}
G.~Barchi, D.~Fontanelli, D.~Macii, and D.~Petri, ``On the accuracy of phasor
  angle measurements in power networks,'' \emph{IEEE Trans. on Instrumentation
  and Measurement}, 2015.

\bibitem{CLMS08}
A.~Carta, N.~Locci, C.~Muscas, and S.~Sulis, ``A flexible gps-based system for
  synchronized phasor measurement in electric distribution networks,''
  \emph{IEEE Trans. on Instrumentation and Measurement}, 2008.

\bibitem{LBFMPM12}
M.~Lixia, A.~Benigni, A.~Flammini, C.~Muscas, F.~Ponci, and A.~Monti, ``A
  software-only ptp synchronization for power system state estimation with
  pmus,'' \emph{IEEE Trans. on Instrumentation and Measurement}, 2012.

\bibitem{DRWP20}
A.~{Dervi\v{s}kadi\'{c}}, R.~{Razzaghi}, Q.~{Walger}, and M.~{Paolone}, ``The
  white rabbit time synchronization protocol for synchrophasor networks,''
  \emph{IEEE Transactions on Smart Grid}, vol.~11, no.~1, pp. 726--738, Jan
  2020.

\bibitem{SBDSF19}
E.~{Shereen}, F.~{Bitard}, G.~{D\'{a}n}, T.~{Sel}, and S.~{Fries}, ``Next steps
  in security for time synchronization: Experiences from implementing ieee 1588
  v2.1,'' in \emph{2019 IEEE International Symposium on Precision Clock
  Synchronization for Measurement, Control, and Communication (ISPCS)}, Sep.
  2019, pp. 1--6.

\bibitem{Jiang2013GPS}
X.~Jiang, J.~Zhang, B.~J. Harding, J.~J. Makela, and A.~D.
  Domı´nguez-Garcı´a, ``Spoofing gps receiver clock offset of phasor
  measurement units,'' \emph{IEEE Transactions on Power Systems}, vol.~28,
  no.~3, pp. 3253--3262, Aug 2013.

\bibitem{Barreto2016}
S.~Barreto, A.~Suresh, and J.~Y.~L. Boudec, ``Cyber-attack on packet-based time
  synchronization protocols: The undetectable delay box,'' in \emph{2016 IEEE
  Int. Instrumentation and Measurement Technology Conf. Proceedings}, May 2016,
  pp. 1--6.

\bibitem{rank1}
S.~Barreto, M.~Pignati, G.~D\'{a}n, J.-Y.~L. Boudec, and M.~Paolone,
  ``Undetectable pmu timing-attack on linear state-estimation by using rank-1
  approximation,'' \emph{IEEE Trans. on Smart Grid}, 2016.

\bibitem{SDBDBP19}
E.~{Shereen}, M.~{Delcourt}, S.~{Barreto}, G.~{D\'{a}n}, J.-Y. {Le Boudec}, and
  M.~{Paolone}, ``Feasibility of time synchronization attacks against pmu-based
  state-estimation,'' \emph{IEEE Transactions on Instrumentation and
  Measurement}, pp. 1--1, 2019.

\bibitem{abur2004power}
A.~Abur and A.~Exposito, \emph{Power system state estimation: theory and
  implementation}.\hskip 1em plus 0.5em minus 0.4em\relax CRC, 2004, vol.~24.

\bibitem{Liu2009}
\BIBentryALTinterwordspacing
Y.~Liu, P.~Ning, and M.~K. Reiter, ``False data injection attacks against state
  estimation in electric power grids,'' in \emph{Proceedings of the 16th ACM
  Conference on Computer and Communications Security}, ser. CCS '09.\hskip 1em
  plus 0.5em minus 0.4em\relax New York, NY, USA: ACM, 2009, pp. 21--32.
  [Online]. Available: \url{http://doi.acm.org/10.1145/1653662.1653666}
\BIBentrySTDinterwordspacing

\bibitem{Kosut2011}
O.~Kosut, L.~Jia, R.~J. Thomas, and L.~Tong, ``Malicious data attacks on the
  smart grid,'' \emph{IEEE Transactions on Smart Grid}, vol.~2, no.~4, pp.
  645--658, Dec 2011.

\bibitem{Dan2010}
G.~D\'{a}n and H.~Sandberg, ``Stealth attacks and protection schemes for state
  estimators in power systems,'' in \emph{2010 First IEEE International
  Conference on Smart Grid Communications}, Oct 2010, pp. 214--219.

\bibitem{Anwar2015}
\BIBentryALTinterwordspacing
A.~Anwar, A.~N. Mahmood, and Z.~Tari, ``Identification of vulnerable node
  clusters against false data injection attack in an ami based smart grid,''
  \emph{Inf. Syst.}, vol.~53, no.~C, pp. 201--212, Oct. 2015. [Online].
  Available: \url{http://dx.doi.org/10.1016/j.is.2014.12.001}
\BIBentrySTDinterwordspacing

\bibitem{Vukovic2012}
O.~Vukovi\'{c}, K.~C. Sou, G.~D\'{a}n, and H.~Sandberg, ``Network-aware
  mitigation of data integrity attacks on power system state estimation,''
  \emph{IEEE Journal on Selected Areas in Communications}, vol.~30, no.~6, pp.
  1108--1118, July 2012.

\bibitem{5751206}
T.~T. {Kim} and H.~V. {Poor}, ``Strategic protection against data injection
  attacks on power grids,'' \emph{IEEE Transactions on Smart Grid}, vol.~2,
  no.~2, pp. 326--333, 2011.

\bibitem{Zhang2013}
Z.~Zhang, S.~Gong, A.~Dimitrovski, and H.~Li, ``Time synchronization attack in
  smart grid: Impact and analysis,'' \emph{IEEE Transactions on Smart Grid},
  vol.~4, no.~1, pp. 87--98, March 2013.

\bibitem{arpan}
A.~{Wehenkel}, A.~{Mukhopadhyay}, J.~{Le Boudec}, and M.~{Paolone}, ``Parameter
  estimation of three-phase untransposed short transmission lines from
  synchrophasor measurements,'' \emph{IEEE Transactions on Instrumentation and
  Measurement}, pp. 1--1, 2020.

\bibitem{LA18}
Y.~Lin and A.~Abur, ``A highly efficient bad data identification approach for
  very large scale power systems,'' \emph{IEEE Transactions on Power Systems},
  2018.

\bibitem{GA15}
M.~Göl and A.~Abur, ``A modified chi-squares test for improved bad data
  detection,'' \emph{IEEE Eindhoven PowerTech}, 2015.

\bibitem{AB19}
Y.~Al-Eryani and U.~Baroudi, ``An investigation on detecting bad data injection
  attack in smart grid,'' \emph{International Conference on Computer and
  Information Sciences (ICCIS)}, 2019.

\bibitem{7839276}
X.~{Fan}, L.~{Du}, and D.~{Duan}, ``Synchrophasor data correction under gps
  spoofing attack: A state estimation-based approach,'' \emph{IEEE Transactions
  on Smart Grid}, vol.~9, no.~5, pp. 4538--4546, 2018.

\end{thebibliography}

\appendix
\subsection{Proof of theorem~\ref{cl:main}}

\begin{proof} 
\begin{itemize}
\item \emph{$S$ non-critical $\rightarrow F^S$ full rank:} Assume that $F^S$ is not full rank and that $|S|=p$, hence there exists a $ y \in\Comps^{p}$ such that $\begin{bmatrix} F^S \end{bmatrix} \begin{pmatrix} y \end{pmatrix} =0$. Thus, by construction $\begin{bmatrix} F \end{bmatrix} \begin{pmatrix} y_1 & ... & y_p  & 0 & ... & 0 \end{pmatrix}^T=0$. Hence, by definition of $F$, vector $v= (H^\dagger H)^{-1}H^\dagger \begin{pmatrix} y \\ 0 \end{pmatrix} \in \Comps^{N}$ is such that $Hv= \begin{pmatrix} y \\ 0 \end{pmatrix}$. Hence, the submatrix of $H$ obtained by removing it's first $p$ rows is not full rank, which implies that $S$ is critical. We have shown that if $F^S$ is not full rank, then $S$ is critical. This is equivalent to showing that if $S$ is non-critical, then $F^S$ is full rank.

\item \emph{$S$ non-critical $\leftarrow F^S$ full rank:}
If $S$ is critical, then there exists a $v \in\Comps^{N}$ such that $Hv=\begin{pmatrix} h_1 v & ... & h_p v & 0 \end{pmatrix}^T$. Because $H$ is full rank, it must be that at least one value $h_i v$ for $1 \leq i \leq p$ is non-zero. Hence, vector $\begin{pmatrix} h_1 v & ... & h_p v & 0 \end{pmatrix}^T$ is in the range of $H$, this means that its orthogonal projection on $Im(H)$ is equal to itself. By definition of $F$ it must be that $\begin{bmatrix} F \end{bmatrix} \begin{pmatrix} h_1 v & \cdots & h_p v & 0 \end{pmatrix}^T=\begin{bmatrix} F^S \end{bmatrix} \begin{pmatrix} h_1 v & \cdots & h_p v  \end{pmatrix}^T=0.$ Therefore, $F^S$ is not full rank. We have shown that if $S$ is critical, then $F^S$ is not full rank. This is equivalent to showing that if $F^S$ is full rank, then $S$ is non-critical.
\end{itemize}
We conclude that $F^S$ is full rank if and only if $S$ is non-critical.
\end{proof}

\subsection{Proof of theorem~\ref{th:p-k}}

\begin{lem}\label{cl:2} 
If the set of measurements with indices in $S$ is critical minimal, then $\mbox{rank}(F^S)=p-1$. 
\end{lem}
\begin{proof} 
If $S$ is critical, then theorem~\ref{cl:main} implies that $F^S$ is not full rank and hence the $p$ columns are dependent: $\mbox{rank}(F^S) \leq p-1.$ Also, if $S$ is critical minimal, then any subset of $p-1$ measurements included in $S$ is non-critical. From theorem~\ref{cl:main}, this means that any submatix $F^{S-1}$, obtained by removing a column of $F^S$, is full rank. Hence, any subset of $p-1$ columns of $F^S$ is independent, this implies that $\mbox{rank}(F^S) \geq p-1.$ Therefore the bound is tight: $\mbox{rank}(F^S) = p-1$.
\end{proof}

\begin{lem}\label{cl:3} 
If $\mbox{rank}(F^S)=p-1$, then the set of measurements with indices in $S$ is critical but not always minimal.
\end{lem}
\begin{proof}
If $\mbox{rank}(F^S)=p-1$, then there exists $x \in \Comps^{p}$ such that $F^Sx=0$ with non-all zero $x_i$ for $1 \leq i \leq p$. By definition of $F$, vector $u=(H^\dagger H)^{-1}H^\dagger \begin{pmatrix} x \\ 0 \end{pmatrix}$ is such that $Hu= \begin{pmatrix}x \\ 0 \end{pmatrix} $. Hence, the set of measurements $\{z_1,...,z_p\}$ is critical. However, the constraint on vector $x$ is that it can have some zero values but it cannot be the zero vector, thus the set $S$ is not necessarily critical minimal. 
\end{proof}
By using the results of Lemmas ~\ref{cl:2} and~\ref{cl:3} on subsets of $S$ of size $p-s+1$ with $0\leq s \leq p$, we obtain that if rank$(F^S)=p-s$, then all subsets of size $p-s+1$ is critical and there is at least one set of size $p-s$ that is non critical. This proves theorem~\ref{th:p-k}.

\subsection{Proof of the Vulnerability Condition for Pairs of Sites using rectangular notations}\label{ap:c}
\begin{theorem}\label{db}
A pair of sites $S_1$ and $S_2$ with non-zero and non-critical measurements that are not vulnerable by themselves, are vulnerable together to undetectable TSAs if and only if there exist an $l \in \mathbb{C}^*$  such that $R^{S_1}z^{S_1}=lR^{S_2}z^{S_2}$.
\end{theorem}

\begin{proof}
Undetectability is equivalent to $r(\Delta z_{\Box})=0$ with $\Delta z=\begin{pmatrix} (u_1-1) z^{S_1} & (u_2-1) z^{S_2} & 0 \end{pmatrix}^T$. Hence, by theorem~\ref{lem:eq}, the vulnerability condition is equivalent to $(u_1-1)R^{S_1}z^{S_1}+(u_2-1)R^{S_2}z^{S_2}=0$. In other words, the pair of sites is vulnerable if and only if there exist an $l=-\frac{u_2-1}{u_1-1} \in \mathbb{C}^*$ such that $|u_2|=|u_1|=1$ and $R^{S_1}z_{\Box}^{S_1}=lR^{S_2}z_{\Box}^{S_2}.$  Note that we can devide by $(u_1-1)$ because $u_1 = 1$ is the non-attack solution. 
\end{proof}

\subsection{Proof of Claim~\ref{cl:err}}
\begin{proof}
Recall the structure of the attack-angle matrix: $$W= \begin{bmatrix} |z_1|^2 f_{11} & \bar{z_1} z_2 f_{12} \\ z_1 \bar{z_2} f_{21} & |z_{2}|^2 f_{22} \end{bmatrix}, $$ where $f_{it}=(F^{S_i})^{\dagger} F^{S_t}$. Using this notation, it was shown in~\cite{rank1} that $IoS_{1,2}^*=\frac{1}{2}+ \frac{|f_{12}|}{2 \sqrt{f_{11} f_{22}}}.$ Our metric $ERR(F^{[S_1,S_2]})$ is equal to $IoS(X)$ and $X=(F^{[S_1,S_2]})^{\dagger}F^{[S_1,S_2]}=\begin{bmatrix} f_{11} & f_{12} \\ f_{21} & f_{22} \end{bmatrix}$. According to~\cite{rank1}, $IoS(X)= \frac{1}{2} + \frac{1}{2} \sqrt{1- 4\frac{\text{det}(X)}{\text{tr}(X)^2}},$ where det$(x)=f_{11}f_{22}-|f_{12}|^2$ and tr$(X)=f_{11}+f_{22}$ are the determinant and trace of $X$, respectively. Hence, 
\begin{align*}IoS(X) &= \frac{1}{2} + \frac{1}{2} \sqrt{1- 4\frac{f_{11}f_{22}-|f_{12}|^2}{(f_{11}+f_{22})^2}} \\
&= \frac{1}{2}+\frac{1}{2} \frac{\sqrt{(f_{11}+f_{22})^2-4f_{11} f_{22}+4|f_{12}|^2}}{f_{11}+f_{22}}.
\end{align*} Denoting $P=f_{11}f_{22} \in \mathbb{R}^+$, $D=|f_{12}| \in \mathbb{R}^+
$ and $S=f_{11}+f_{22}$, we have 
\begin{align*}
IoS_{1,2}^* &=\frac{1}{2}+ \frac{1}{2} \sqrt{\frac{D^2}{P}} \text{ and } \\
IoS(X) &=\frac{1}{2}+\frac{1}{2}\sqrt{\frac{S^2-4P+4D^2}{S^2}}.\\
\end{align*}
Given $P$ and $S$, values $f_{11}$ and $f_{22}$ exist if and only if $S^2 \geq 4P$.
Define function $\phi (x)= \frac{x-4P+4D^2}{x} = 1-4 \frac{P-D^2}{x} \leq 1$ because $P \geq D^2$. Because $S^2 \geq 4P$, the lowest value of the definition domain of $\phi (x)$ is $4P$: $\phi (4P)=\frac{D^2}{P}$. The derivative $\phi '(x)= \frac{P-D^2}{x^2}$ being positive implies that $\frac{S^2-4P+4D^2}{S^2} \geq \frac{D^2}{P}$, hence $1 \geq IoS(X) \geq IoS_{1,2}^*.$ The inequalities are equalities if $D^2=P$. If $P$ is much larger than $D^2$, then $IoS_{1,2}^*$ is much smaller than $1$ and if $S$ is much larger than $2 \sqrt{P}$, then $IoS(X)$ is close to $1$. In that case, an attack is feasible, $IoS_{1,2}^*$ doest not enable the detection of the vulnerability but $IoS(X)=ERR(F^{[S_1,S_2]})$ does.
This happens when $S >> 2 \sqrt{P}$, thus when 
\begin{align*}
f_{11}+f_{22} & >> 2 \sqrt{f_{11} f_{22}} \\
(f_{11}+f_{22})^2 &>> 4 f_{11} f_{22} \\
f_{11}^2-2f_{11}f_{22}+f_{22}^2 =(f_{11}-f_{22})^2 &>>0 \\
f_{11} >> f_{22} & \text{ or } f_{22} >>f_{11} 
\end{align*}
In other words, if the values of $F^{S_1}$ are much larger than the values of $F^{S_2}$ or vice-versa, then $ERR(F^{[S_1,S_2]})$ is close to $1$ but $IoS_{1,2}^*$ is not.
 \end{proof}

\end{document}